\documentclass[11pt]{article}
\usepackage{amssymb,amsmath,amsthm,sectsty,url,ifpdf}
\usepackage[letterpaper,hmargin=0.972in,vmargin=1.0in]{geometry}

\usepackage[usenames,dvipsnames]{xcolor}
\usepackage[colorlinks=true,linkcolor=NavyBlue,citecolor=teal]{hyperref}
\usepackage{boxedminipage}
\usepackage[boxed]{algorithm}
\usepackage{cleveref,aliascnt}

\usepackage[normalem]{ulem}
\usepackage[boxed]{algorithm}
\usepackage[labelfont=bf,small]{caption}
\floatname{algorithm}{Figure}
\usepackage{multicol}
\usepackage{graphicx}
\graphicspath{ {images/} }

\theoremstyle{plain}
\newtheorem{theorem}{Theorem}[section]
\newtheorem{lemma}[theorem]{Lemma}
\newtheorem{corollary}[theorem]{Corollary}

\newtheorem{claim}[theorem]{Claim}

\theoremstyle{definition}
\newtheorem{definition}[theorem]{Definition}
\newtheorem{remark}[theorem]{Remark}

\newcommand{\abs}[1]{\left|#1\right|}

\def\E{\mathop{\mathbb E}}
\newcommand{\Eparen}[1]{\E \left [#1 \right ]}

\newcommand{\N}{\mathbb N}
\newcommand{\R}{\mathbb R}

\newcommand{\B}{\{ 0,1 \}}

\newcommand{\HH}{\mathcal{H}}

\newcommand{\BAD}{\mathsf{BAD}}
\newcommand{\err}{\textsf{err}}

\newcommand{\poly}{\textsf{poly}}

\newcommand{\leps}{\log{\frac{1}{\varepsilon}}}

\newcommand{\paren}[1]{\left(#1\right)}

\newcommand{\logchoose}[2]{\log{ {#1 \choose #2}}}
\newcommand{\logp}[1]{\log{\paren{#1}}}
\newcommand{\eps}{\varepsilon}

\newcommand{\till}{,\ldots,}

\newcommand{\ie}  {i.e.,\ }
\newcommand{\eg}  {e.g.,\ }

\newcommand{\etal}{{et~al.\ }}
\newcommand{\etalcite}[1]{{et~al.~\cite{#1}}}
\newcommand{\negl}{{\mathsf{neg}}}
\newcommand{\ith}[1]{{#1}\textsuperscript{th}}

\newcommand{\PRF}{\textsf{PRF}}
\newcommand{\PRP}{\textsf{PRP}}

\newcommand{\Algorithm}[1]{\textsf{#1}}

\newcommand{\Predict}{\Algorithm{Attack}}
\newcommand{\PredictUnsteady}{\Algorithm{Attack}}

\newcommand{\BB}{{\mathbf{B}}}
\newcommand{\M}{{M}}

\newcommand{\Challenge}{{\sf Challenge}}

\newcommand{\myand}{{\;\;\wedge\;\;}}
\newcommand{\Xc}{\mathcal{X}}

\title{
Bloom Filters in Adversarial Environments}
\author{
 Moni Naor\thanks{Weizmann Institute of Science. Email:
    \texttt{moni.naor@weizmann.ac.il}.
    Supported in part by a grant from the I-CORE Program of the
    Planning and Budgeting Committee, the Israel Science Foundation, BSF and the Israeli Ministry of Science and Technology. Incumbent of the Judith Kleeman Professorial Chair.}
\and
Eylon Yogev\thanks{Weizmann Institute of Science. Email:
    \texttt{eylon.yogev@weizmann.ac.il}.
    Supported in part by a grant from the I-CORE Program of the
    Planning and Budgeting Committee, the Israel Science Foundation, BSF and the Israeli Ministry of Science and Technology.
     }}

\date{}

\begin{document}

\maketitle

\begin{abstract}
Many efficient data structures use randomness, allowing them to 
improve upon deterministic ones. Usually, their efficiency and
correctness are analyzed using probabilistic tools under the assumption that 
the inputs and queries are \emph{independent} of the internal randomness of 
the data structure. In this work, we consider data structures in a more robust 
model, which we call the \emph{adversarial model}. Roughly speaking, this 
model allows an adversary to choose inputs and queries \emph{adaptively} 
according to previous responses. Specifically, we consider a data structure 
known as ``Bloom filter'' and prove a tight connection between Bloom filters 
in this model and cryptography.

A Bloom filter represents a set $S$ of elements approximately, by using fewer
bits than a precise representation.
The price for succinctness is allowing some errors: for any $x \in S$ it
should always
answer `Yes', and for any $x \notin S$ it should answer `Yes' only with
small probability.

In the adversarial model, we consider both efficient adversaries (that run in 
polynomial time) and computationally unbounded adversaries that are only 
bounded in the number of queries they can make. For computationally bounded 
adversaries, we show that non-trivial (memory-wise) Bloom filters exist if and 
only if one-way functions exist. For unbounded adversaries we show that there 
exists a Bloom filter for sets of size $n$ and error $\eps$, that is secure 
against $t$ queries and uses only $O(n \leps+t)$ bits of memory. In 
comparison, $n\leps$ is the best possible under a non-adaptive adversary.
\end{abstract}

\newpage
\tableofcontents
\newpage

\section{Introduction}
Data structures are one of the most fundamental objects in Computer Science. 
They
provide means to organize a large amount of data such that it can be
queried efficiently. In general, constructing efficient data structures is key 
to designing efficient algorithms. Many efficient data structures use 
randomness, a resource that allows them to 
bypass lower bounds on deterministic ones. In these cases, their efficiency 
and correctness are analyzed in expectation or with high probability.

To analyze randomized data structures, one must first define the underlying 
model of the analysis. Usually, the model assumes that 
the inputs (equivalently, the queries) are \emph{independent} of the internal 
randomness of 
the data structure. That is, the analysis is of the form: For any
sequence of inputs, with high probability (or expectation) over its internal
randomness, the data structure will yield a correct answer. This model is
reasonable in a situation where the adversary picking the inputs gets no
information about the internal state of the data structure or about the random 
bits actually used (in particular, the 
adversary does 
not get the responses on previous inputs).\footnote{This does not include Las 
Vegas type data structures, where the output is always correct, and the 
randomness only affects the running time.}

In this work, we consider data structures in a more robust model, which we 
call the \emph{adversarial model}. Roughly speaking, this 
model allows an adversary to choose inputs and queries \emph{adaptively} 
according to previous responses. That is, 
the analysis is of the form: With high probability over 
the internal randomness of the data structure, for any adversary adaptively 
choosing a sequence of inputs, the response to a single query will be correct. 
Specifically, we consider a data structure known as ``Bloom filter'' 
and prove a tight connection between Bloom filters in this model and 
cryptography: We show that Bloom filters in an adversarial model exist if and 
only if one-way functions exist.

\paragraph{Bloom Filters in Adversarial Environments.}
The approximate set membership problem deals with succinct representations of 
a set $S$ of elements from a large universe $U$, where the price for 
succinctness is allowing some errors. A data structure
solving this problem is required to answer queries in the following manner: 
for any $x \in S$ it should always
answer `Yes', and for any $x \notin S$ it should answer `Yes' only with
small probability. False responses for $x \notin S$ are called \emph{false 
positive} errors.

The study of the approximate set membership problem began with
Bloom's 1970 paper \cite{Bloom70}, introducing the so-called ``Bloom 
filter'',
which provided a simple and elegant solution to the problem. (The term ``Bloom
filter" may refer to Bloom's original construction, but we use it to
denote any construction solving the problem.) The two major
advantages of Bloom filters are: (i) they use significantly less memory (as
opposed to storing $S$ precisely) and (ii) they have very fast query time (even
constant query time). Over the years, Bloom filters have been found to be
extremely useful and practical in various areas. Some primary examples are
distributed systems \cite{ZhuJW04}, networking \cite{DharmapurikarKSL04},
databases~\cite{Mullin90}, spam filtering~\cite{YanC06,LiZ06}, web 
caching~\cite{FanCAB00}, streaming algorithms~\cite{NaorY15,DengR06} and 
security~\cite{ManberW94,ZhangG08}. For a survey about Bloom
filters and their applications see~\cite{BroderM03} and  a 
more recent one~\cite{TarkomaRL12}.

Following Bloom's original construction many generalizations and variants have 
been proposed and extensively
analyzed, providing better tradeoffs between memory consumption, error 
probability and 
running time, see e.g.,~\cite{ChazelleKRT04,PutzeSS09,PaghPR05,ArbitmanNS10}. 
However, as 
discussed, all known constructions of
Bloom filters work under the assumption that the input query $x$ is fixed, and 
then the probability of an error occurs over the randomness of the 
construction.
Consider the case where the query results are made public. What happens
if an adversary chooses the next query according to the responses of 
previous ones? Does the bound on the error probability still hold? The 
traditional analysis of Bloom filters is no longer sufficient, and stronger 
techniques are required.

Let us demonstrate this need with a concrete scenario. Consider a system where 
a Bloom 
filter representing a {\em white
list} of email addresses is used to filter spam mail. When an email message is
received, the sender's address is checked against the Bloom filter, and if the
result is negative, it is marked as spam. Addresses not on the white list have 
only a small probability of being a false positive and thus not marked as 
spam. In this case, the results of the queries are public, as an attacker 
might check whether his emails are marked as spam (\eg spam his personal email 
account and see if the messages are being 
filtered). In this case, each query translates to opening a new email account, 
which might be costly. Moreover, an email address can be easily blocked once 
abused. Thus, the goal of an attacker is to find a large bulk of email 
addresses using a small number of queries. Indeed, the attacker, after a short 
sequence of queries, might be able to find 
a {\em large} bulk of email addresses (much larger than the number of queries) 
that are not marked as spam although they are not in the white list. Thus, 
bypassing the security of the system and flooding users with spam mail.

As another example application, Bloom filters are often used for holding the 
contents 
of a cache. For instance, a web proxy holds, on a (slow) disk, a cache 
of locally available web pages. To improve performance, it maintains in (fast) 
memory a Bloom filter representing all addresses in the cache. 
When a user queries for a web page, the proxy first checks the Bloom filter to 
see if the page is available in the cache, and only then does it
search for the web page on the disk. A false positive is translates to  
unsuccessful cache access, that is, a slow disk lookup. In the 
standard analysis, one would set the error to be small such that cache misses 
happen very rarely (\eg 
one in a thousand requests). However, by timing the results of the proxy, an 
adversary might learn the responses of the Bloom filter, enabling her to find 
false positives and cause 
an unsuccessful cache access for almost every query and, eventually, causing a 
Denial of 
Service (DoS) attack. The adversary cannot use a false 
positive more than once, as after each unsuccessful cache access the web page 
is added to the cache.
These types of attacks are applicable in many different 
systems where Bloom filters are integrated (\eg \cite{PappasKVKMCGKB14}).

Under the adversarial model, we construct Bloom filters that are resilient to 
the above attacks. We consider 
both efficient adversaries (that 
run in polynomial time) and computationally unbounded adversaries that are 
only bounded in the number of queries they can make. We define a Bloom filter 
that maintains its error probability in this setting and say it is 
\emph{adversarial resilient} (or just resilient for shorthand).

The security of an adversarial resilient Bloom filter is defined in terms of a 
game (or an experiment) with an
adversary. The adversary is allowed to choose the set $S$, make a sequence of 
$t$ 
adaptive queries
to the Bloom filter and get its responses. Note that the adversary has only
oracle access to the Bloom filter and cannot see its internal memory
representation. Finally, the adversary must output
an element $x^*$ (that was not queried before) which she believes is a false
positive. We say that a Bloom filter is $(n,t,\eps)$-adversarial resilient
if when initialized over sets of size $n$ then after $t$ queries the
probability of $x^*$ being a false positive is at most $\eps$. If a Bloom 
filter is resilient $t$ queries, for any $t$ that is bounded by a polynomial in 
$n$ we 
say it is \emph{strongly resilient}.

A simple construction of a strongly resilient Bloom filter (even against
computationally unbounded adversaries) can be achieved by storing $S$
precisely. Then, there are no false positives at all and no adversary can find
one. The drawback of this solution is that it requires a large amount of
memory, whereas Bloom filters aim to reduce the memory usage. We are interested
in Bloom filters that use a small amount of memory but remain nevertheless,
resilient.

\subsection{Our Results}
We introduce the notion of {\em adversarial-resilient
Bloom filter} and show several possibility results (constructions of
resilient Bloom filters) and impossibility results (attacks against any
Bloom filter) in this context. The precise definitions and the model we consider are given in Section \ref{sec:model}.

\paragraph{Lower bounds.}
Our first result is that adversarial-resilient Bloom filters
against computationally bounded adversaries that are non-trivial (\ie they 
require less space than the amount of space it takes to store the elements 
explicitly) must use one-way functions. That is, we show that if one-way 
functions do not exist then any Bloom filter can be `attacked' with high 
probability.

\begin{theorem}[Informal]\label{thm:informal1}
Let $\BB$ be a \emph{non-trivial} Bloom filter. If $\BB$ is strongly
resilient against computationally bounded adversaries, then one-way functions
exist.
\end{theorem}
Actually, we show a trade-off between the amount of memory used by 
the Bloom filter and the number of queries performed by the adversary. Carter 
\etalcite{CarterFGMW78} proved a lower 
bound on the amount of memory required by a Bloom filter. To construct a Bloom 
filter for sets of size $n$ and error rate $\eps$ one must use (roughly) 
$n\leps$ bits of memory (and this is tight). Given a Bloom filter that uses 
$m$ bits of memory we get a lower bound for its error rate $\eps$ and thus a 
lower bound for the (expected) number of false positives. The smaller $m$ is, 
the larger the number of false positives is, and we prove that the adversary 
can perform fewer queries.

Bloom filters consist of two algorithms: an initialization algorithm that gets 
a set and outputs a compressed representation of the set, and a membership 
query algorithm that gets a representation and an input. Usually, Bloom filters 
have a {\em randomized} initialization algorithm but a {\em deterministic} 
query algorithm that does not change the representation. We say that such Bloom 
filters have a ``steady representation''. However, in some cases, a randomized 
query algorithm can make the Bloom filter more powerful (see 
\cite{ErlingssonPK14} for such an example). Specifically, it might incorporate 
differentially private~\cite{DworkMNS06} algorithms in order to 
protect the internal memory from leaking. Differentially private algorithms are 
designed to protect a private database against adversarial and also adaptive 
queries from a data analyst. One might hope that such techniques can eliminate 
the need for one-way functions in order to construct resilient Bloom filters. 
Therefore, we consider also Bloom filters with ``unsteady representation'': 
where the query algorithm is randomized {\em and} can change the underlying 
representation on each query. We extend our results (\Cref{thm:informal1}) 
to handle Bloom filters with unsteady representations, which proves that any 
such approach 
cannot gain additional security. The proof of the theorem (both the steady and 
unsteady case) appears in Section \ref{sec:essential}.

\paragraph{Constructions.}
In the other direction, we show that using one-way functions one
can construct a strongly resilient Bloom filter. Actually, we show that
one can transform any Bloom filter to be strongly resilient with almost exactly
the same memory requirements and at a cost of a single evaluation of a 
pseudorandom permutation\footnote{A pseudorandom permutation is family
of functions that a random function from the family cannot be distinguished 
from a truly random permutation by any 
polynomially bounded adversary making queries to the function. It models a 
block cipher (See Definition 
\ref{def:pseudorandomP}).} (which can be constructed using one-way 
functions). Specifically, in Section \ref{sec:sufficient} we show:

\begin{theorem}\label{thm:informal2}
	Let $\BB$ be an $(n,\eps)$-Bloom filter using $m$ bits of memory. If
	pseudorandom permutations exist, then for security parameter 
	$\lambda$ there exists a negligible function\footnote{A 
		function $\negl\colon\N\to\R^+$ is \emph{negligible} if for every 
		constant 
		$c > 0$, there exists an integer $N_c$ such that $\negl(n) < n^{-c}$ 
		for 
		all $n > N_c$.} $\negl(\cdot)$ and an 
	$(n,\eps+\negl(\lambda))$-strongly resilient Bloom filter that uses 
	$m'=m+\lambda$ bits of memory.
\end{theorem}

In practice, Bloom filters are used when performance is crucial, and extremely 
fast implementations are required. This raises implementation difficulties 
since cryptographically secure functions rely on relatively heavy computation. 
Nevertheless, we provide an implementation of an adversarial resilient Bloom 
filter that is provably secure under the hardness of AES and is essentially 
as fast as any other implementation of insecure Bloom filters. Our 
implementation 
exploits the AES-NI\footnote{Advanced Encryption Standard Instruction Set.} 
instruction set  that is embedded in most modern CPUs and provides hardware 
acceleration of the AES encryption and decryption algorithms \cite{Gueron2009}. 
See Appendix \ref{sec:implementation} for more details.

In the context of unbounded adversaries, we show a positive result. For a set
of size $n$ and an error probability of $\eps$ most constructions use about
$O(n\leps)$ bits of memory. We construct a resilient Bloom
filter that does not use
one-way functions, is resilient against $t$ queries, uses
$O(n\leps + t)$ bits of memory, and has query time $O(\leps)$.

\begin{theorem}\label{thm:informal3}
There exists an $(n,t,\eps)$-resilient
Bloom filter (against unbounded adversaries) for any $n,t \in \N$, and $\eps>0$ 
that uses $O(n\leps + t)$ bits of memory and has linear setup time and $O(1)$ 
worst-case query time.
\end{theorem}

\subsection{Related Work}
One of the first works to consider an adaptive adversary that chooses queries 
based on the response of the data structure is by Lipton and 
Naughton~\cite{LiptonN93}, where
adversaries that can measure the {\em time} of specific operations in a dictionary were addressed. They
showed how such adversaries can be used to attack hash tables. Hash tables 
have some method for dealing with collisions. An adversary that can measure 
the 
time of an insert query can determine whether there was a collision and might 
figure out the precise hash function used. She can then choose the next 
elements to insert accordingly, increasing the probability of a collision and 
hurting the  overall performance.

Mironov~\etalcite{MironovNS11} considered the model of {\em sketching} in an 
adversarial environment. The model consists of several honest parties that are 
interested in computing a joint function in
the presence of an adversary. The adversary chooses the inputs of the honest 
parties based on the shared random string. These inputs are provided to
the parties in an on-line manner, and each party incrementally updates a compressed sketch of
its input. The parties are not allowed to communicate, they do not share any secret information,
and any public information they share is known to the adversary in advance. Then, 
the parties engage in a protocol in order to evaluate the function on their current 
inputs using only the compressed sketches. Mironov \etal construct explicit and 
efficient (optimal) protocols for two fundamental problems: testing equality of two 
data sets and approximating the size of their symmetric difference.

In a more recent work, Hardt and Woodruff \cite{HardtW13} considered linear
sketch algorithms in a similar setting. They consider an adversary that can
adaptively choose the inputs according to previous evaluations of the sketch.
They ask whether linear sketches can be robust to adaptively chosen inputs.
Their results are negative: They showed that no linear sketch approximates the
Euclidean norm of its input to within an arbitrary multiplicative
approximation factor on a polynomial number of adaptively chosen inputs.

One may consider adversarial resilient Bloom filters in the framework of 
computational learning theory. The task of the adversary is to
learn the private memory of the Bloom filter in the sense that it is able to predict on which elements the Bloom filter outputs a false positive.
The connection between learning and cryptographic assumptions has been 
explored before (already in his 1984 paper introducing the PAC model Valiant's 
observed that the nascent pseudorandom functions imply hardness of 
learning~\cite{Valiant84}). In particular, Blum \etalcite{BlumFKL93}
showed how to construct several cryptographic primitives (pseudorandom bit generators, one-way functions and private-key
cryptosystems) based on certain assumptions on the difficulty of learning. The 
necessity of one-way functions for several cryptographic primitives has been 
shown in \cite{ImpagliazzoL89}.

\section{Model and Problem Definitions}\label{sec:model}
In our model, we are given a universe $U = [u]$ of elements, and a subset $S 
\subset U$ of size $n$. For simplicity of presentation, we consider mostly the 
static 
problem, where the set is fixed throughout the lifetime of the data structure. In the dynamic setting, the Bloom filter is initially empty, and the user can add elements to the set in between queries. We note that the lower bounds imply the same bounds for the dynamic case and the cryptographic 
upper bound (\Cref{thm:no-owf}) actually works in the dynamic case as well.

A Bloom filter is a data structure $\BB=(\BB_1,\BB_2)$ composed of a setup 
algorithm $\BB_1$ (or ``build'') and a query algorithm $\BB_2$ (or ``query'').
The setup algorithm $\BB_1$ is
randomized, gets as input a set $S$, and outputs $\BB_1(S)=\M$ which is a 
compressed representation of the set $S$. To denote the representation $\M$ on 
a set $S$ with random
string $r$ we write $\BB_1(S;r)=M^S_r$; its size in bits is denoted as
$|\M^S_r|$.

The query algorithm answers membership queries to $S$ given the compressed
representation $\M$. That is, it gets an input $x$ from $U$ and
answers with 0 or 1. (The idea is that the answer is 1 only if $x \in S$, but
there may be errors.)
 Usually, in the literature, the query algorithm is
deterministic and cannot change the representation. In this case, we say $\BB$
has a \emph{steady representation}.
However, we also consider Bloom filters where their query algorithm is
\emph{randomized} and can change the representation $\M$ after each query. In
this case, we say that
$\BB$ has an \emph{unsteady representation}. We define both variants.

\begin{definition}[Steady-representation Bloom filter]\label{def:steady}
Let $\BB=(\BB_1,\BB_2)$ be a pair of polynomial-time algorithms where $\BB_1$ 
is a randomized algorithm that gets as input a set $S$ and outputs a 
representation, and $\BB_2$ is a deterministic algorithm that gets as input a 
representation and a query element $x \in U$. We say that $\BB$ is an 
$(n,\eps)$-Bloom filter (with a steady
representation) if for all sets $S$ of size $n$ in a suitable universe $U$ it 
holds that:
\begin{enumerate}
\item Completeness: For any $x \in S$: $\Pr[\BB_2(\BB_1(S),x) = 1] = 1$
\item Soundness: For any $x \notin S$: $\Pr[\BB_2(\BB_1(S),x) = 1] \le \eps$,
\end{enumerate}
where the probabilities are over the setup algorithm $\BB_1$.
\end{definition}

\paragraph{False Positive and Error Rate.} Given a representation $\M$ of 
$S$, if $x \notin S$ and $\BB_2(M,x)=1$ we say that $x$ is a \emph{false 
positive} with respect to $\M$. Moreover, if $\BB$ is an $(n,\eps)$-Bloom 
filter then we say that 
$\BB$ has \emph{error rate} at most $\eps$.

Definition \ref{def:steady} considers only a single fixed input
$x$ and the probability is taken over the randomness of $\BB$. We want to give
a stronger soundness requirement that considers a sequence of inputs
$x_1,x_2,\ldots,x_t$ that is not fixed but chosen by an adversary, where the
adversary gets the responses of previous queries and can adaptively choose the
next query accordingly. If the adversary's probability of finding a false
positive $x^*$ that was not queried before is bounded by $\eps$, then we say
that $\BB$ is an
$(n,t,\eps)$-resilient Bloom filter. This notion is defined in the challenge
$\Challenge_{A,t}$ which is described below. In this challenge, the 
polynomial-time adversary $A=(A_1,A_2)$ consists of two parts: $A_1$ is 
responsible for choosing a set $S$. 
Then, $A_2$ get $S$ as input, and its goal is to find a false positive $x^*$, 
given only oracle access to a Bloom filter initialized with $S$. The adversary 
$A$ succeeds if $x^*$ is not among the queried elements and is a false 
positive. We measure the success probability of $A$ with respect to the 
randomness in $\BB_1$ and in $A$.
 
Note that in this case, the
setup phase of the Bloom filter and the adversary get an additional input, 
which is the number $\lambda$ (given in unary as $1^\lambda$). This is called 
the {\em security parameter}. Intuitively, this is like the length of a 
password, that is, as it increases the security gets stronger. More formally, 
it enables the running time of the Bloom filter to be polynomial in $\lambda$ 
and thus the error $\eps=\eps(\lambda)$ can be a function of $\lambda$. 
From here on, we always assume that $\BB$ has this format, however, sometimes 
we omit this additional parameter from the writing when clear from the context.
For a steady representation Bloom filter
we define:

\begin{definition}[Adversarial-resilient Bloom filter with a steady
representation]\label{def:adversarialsteady}
Let $\BB=(\BB_1,\BB_2)$ be an $(n,\eps)$-Bloom filter with a steady 
representation (see Definition \ref{def:steady}). We say that $\BB$ is an 
$(n,t,\eps)$-adversarial resilient Bloom
filter (with a steady representation) if for any probabilistic polynomial-time 
adversary $A=(A_1,A_2)$ for all large enough $\lambda \in \N$ it 
holds that:
\begin{itemize}
	\item Adversarial Resilient:
	$\Pr[\Challenge_{A,t}(\lambda)=1] \le \eps$,
\end{itemize}
where the probabilities are taken over the internal randomness of $\BB_1$ and 
$A$. The random variable $\Challenge_{A,t}(\lambda)$ 
is the outcome of the following algorithm (see also Figure \ref{fig:def}):
\paragraph{\uline{$\Challenge_{A,t}(\lambda)$:}}
\begin{enumerate}
\item $S \gets A_1(1^{\lambda+n\log{u}})$ where $S \subset U$, and $|S| = n$.
\item $\M \leftarrow \BB_1(1^{\lambda+n\log{u}},S)$.
\item $x^* \leftarrow A_2^{\BB_2(\M,\cdot)}(1^{\lambda+n\log{u}},S)$ where 
$A_2$ 
performs at most $t$ adaptive queries $x_1 \till x_t$ to $\BB_2(\M,\cdot)$.
\item If $x^* \notin S \cup \{x_1 \till x_t\}$ and $\BB_2(\M,x^*)=1$ output 1, otherwise output 0.
\end{enumerate}
\end{definition}
\begin{figure}[h]
\centerline{\includegraphics[scale=0.5]{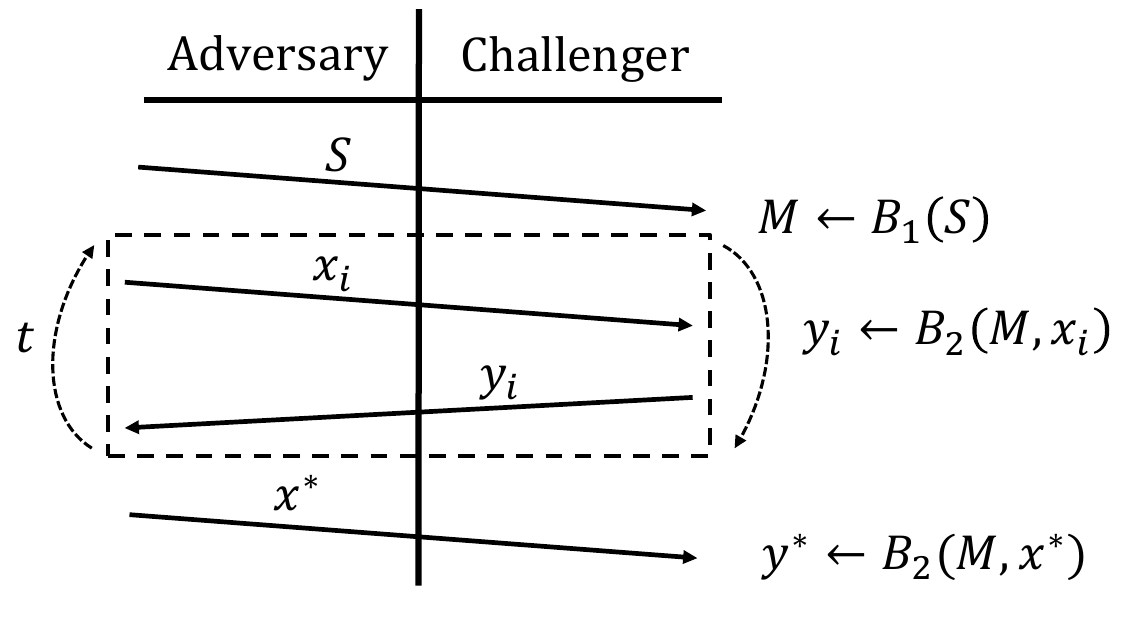}}
\caption{An illustration of Definition \ref{def:adversarialsteady}.
The adversary wins if $y^*=1 \myand x^* \notin S \cup \{x_1,\ldots,x_t\}$.}
\label{fig:def}
\end{figure}

\paragraph{Unsteady representations.}
When the Bloom filter has an unsteady representation, then the algorithm $\BB_2$
is randomized and moreover can change the representation $\M$. That is,
$\BB_2$ is a query algorithm that outputs the response to the query as well as  a new
representation. That is, the data structure has an internal state that might be 
changed with each query and the user can perform queries to this data 
structure. Thus, the user or the adversary do not interact
directly with $\BB_2(\M,\cdot)$ but with an interface $Q(\cdot)$ 
(initialized with some $\M^*$) that on query $x$ updates its 
representation $\M$ and outputs only the
response to the query (i.e., it cannot issue successive queries to the same 
memory representation but to one that keeps changing). Formally, $Q(\cdot)$ 
is with $\M^*$ and at each point in time has a memory $M$. Then, on input $x$ 
is acts as follows:
\paragraph{\uline{The interface $Q(x)$:}}
\begin{enumerate}
\item $(\M',y) \leftarrow \BB_2(\M,x)$.
\item $\M \leftarrow \M'$.
\item Output $y$.
\end{enumerate}
We define a variant of the above interface, denoted by $Q(x;x_1,\ldots,x_t)$ 
(initialized with $\M^*$) 
as 
first performing the queries $x_1,\ldots,x_t$ and then performing the query $x$ 
and output the result on $x$. Formally, we define:
\paragraph{\uline{$Q(x;x_1,\ldots,x_t)$}:}
\begin{enumerate}
\item For $i=1\ldots t$ query $Q(x_i)$.
\item Output $Q(x)$.
\end{enumerate}
We define an analog of the original Bloom filter for unsteady 
representations and then define an adversarial resilient one.

\begin{definition}[Bloom filter with an unsteady 
representation]\label{def:unsteady}
Let $\BB=(\BB_1,\BB_2)$ be a pair of 
probabilistic polynomial-time algorithms such that $\BB_1$ gets as input the 
set $S$ of size $n$ and outputs a representation $\M_0$, and $\BB_2$ gets as 
input a 
representation and query $x$ and outputs a new representation and a response 
to the query. Let $Q(\cdot)$ be the process initialized with $\M_0$. We say 
that $\BB$ is an $(n,\eps)$-Bloom
filter (with an unsteady representation) if for any such set $S$ the following two conditions hold:
\begin{enumerate}
	\item Completeness: For any $x \in S$, for any $t \in \N$ and for any 
	sequence of queries $x_1,x_2,\ldots,x_t$ we have that 
	$\Pr[Q(x;x_1,\ldots,x_t) = 1] = 1$.
	\item Soundness: For any $x \notin S$, for any $t \in \N$ and for any 
		sequence of queries $x_1,x_2,\ldots,x_t$ we have that 
		$\Pr[Q(x;x_1,\ldots,x_t) = 1] \le \eps$,
\end{enumerate}
where the probabilities are taken over the internal randomness of $Q$.
\end{definition}
The above definition is for fixed query sequences, where the next definition is 
for adaptively chosen query sequences.

\begin{definition}[Adversarial-resilient Bloom filter with an unsteady
representation]\label{def:adversarialunsteady}
Let $\BB=(\BB_1,\BB_2)$ be an $(n,\eps)$-Bloom filter with an unsteady 
representation (see Definition \ref{def:unsteady}). We say that $\BB$ is an 
$(n,t,\eps)$-adversarial resilient Bloom filter (with an unsteady 
representation) if for any probabilistic polynomial-time adversary 
$A=(A_1,A_2)$ for all large enough $\lambda \in \N$ it holds 
that:
\begin{itemize}
	\item Adversarial Resilient:
	$\Pr[\Challenge_{A,t}(\lambda)=1] \le \eps$,
\end{itemize}
where the probabilities are taken over the internal randomness of $\BB_1, \BB_2$ and
$A$ and where the random variable $\Challenge_{A,t}(\lambda)$ is the outcome of
the following process:
\paragraph{\uline{$\Challenge_{A,t}(\lambda)$:}}
\begin{enumerate}
\item $S \gets A_1(1^{\lambda+n\log{u}})$ where $S \subset U$, and $|S| \le n$.
\item $\M_0 \leftarrow \BB_1(S,1^{\lambda+n\log{u}})$.
\item Initialize $Q(\cdot)$ with $\M_0$.
\item $x^* \leftarrow A_2^{Q(\cdot)}(1^{\lambda+n\log{u}},S)$ where $A_2$ 
performs at
most $t$ adaptive queries $x_1 \till x_t$ to $Q(\cdot)$.
\item If $x^* \notin S \cup \{x_1 \till x_t\}$ and $Q(x^*)=1$ output 1, otherwise output 0.
\end{enumerate}
\end{definition}
If $\BB$ is not $(n,t,\eps)$-resilient then we say there exists an adversary
$A$ that can $(n,t,\eps)$-attack $\BB$.	
\noindent If $\BB$ is resilient for any
polynomial number of queries we say it is \emph{strongly resilient}:

\begin{definition}[Strongly resilient]
For a security parameter $\lambda$, we say that $\BB$ is an 
$(n,\eps)$-strongly resilient Bloom filter, if for any polynomial 
$t=t(\lambda,n)$ it holds that $\BB$ is an $(n,t,\eps)$-adversarial resilient 
Bloom filter.
\end{definition}

\begin{remark}[Access to the set $S$]
Notice that in Definitions \ref{def:adversarialsteady} and 
\ref{def:adversarialunsteady} the adversary $A_1$ chooses the set $S$, and then 
the adversary $A_2$ gets the 
set $S$ as an additional input. This strengthens the definition of the 
resilient Bloom filter such that even given the set $S$ it is hard to find 
false positives. An alternative definition might be to not give the adversary 
$A_2$ the set. However, our results of 
\Cref{thm:informal1} hold even if the 
adversary does not get the set. That is, the algorithm that predicts a false 
positive makes no use of the set $S$. Moreover, the construction in 
\Cref{thm:informal2} holds in both cases, even against adversaries that do get 
the set.
\end{remark}

An important parameter is the memory use of a Bloom filter $\BB$. We say $\BB$
uses $m=m(n,\lambda,\eps)$ bits of memory if for any set $S$ of size $n$
the largest representation is of size at most $m$. The desired properties of
Bloom filters is to have $m$ as small as possible and to answer membership
queries as fast as possible. Let $\BB$ be a $(n,\eps)$-Bloom filter that uses
$m$ bits of memory. Carter~\etalcite{CarterFGMW78} (see also 
\cite{DietzfelbingerP08} for further details) proved a lower bound on the
memory use of any Bloom filter showing that if $u > n^2/\eps$ then $m \ge 
n\leps$ (or written
equivalently as $\eps \ge 2^{-\frac{m}{n}}$). This leads us to define the 
minimal error of
$\BB$.

\begin{definition}[Minimal error]\label{def:minimal_error}
Let $\BB$ be an  $(n,\eps)$-Bloom filter that uses $m$ bits of
memory. We say that $\eps_0=2^{-\frac{m}{n}}$ is the minimal error of $\BB$.
\end{definition}

Note that using the lower bound of \cite{CarterFGMW78} we get that for any 
$(n,\eps)$-Bloom
filter its minimal error $\eps_0$ always satisfies $\eps_0 \le \eps$. For 
technical reasons, we will have a slightly different condition on the size of 
the universe, and we require that $u =\Omega(m/\eps_0^2)$. 
Moreover, if $u$ is super-polynomial in $n$, and $\eps$ is 
negligible in $n$ then any polynomial-time adversary has 
only negligible chance in finding any false positive, and again we say that 
the Bloom filter is trivial.

\begin{definition}[Non-trivial Bloom filter]\label{def:non-trivial}
Let $\BB$ be an $(n,\eps)$-Bloom filter that uses $m$ bits of memory and let
$\eps_0$ be the minimal error of $\BB$ (see Definition
\ref{def:minimal_error}). We say that $\BB$ is {\em non-trivial} if for all 
constants $a>0$ it holds that $u > \frac{a \cdot m}{\eps_0^2}$ and there exists 
a 
constant $c$ such that $\eps_0 > 
\frac{1}{n^c}$.
\end{definition}
\section{Our Techniques}\label{sec:tech}
\subsection{One-Way Functions and Adversarial Resilient Bloom
Filters}
We present the main ideas and techniques of the equivalence of adversarial
resilient Bloom filters and one-way functions (\ie the proof of
\Cref{thm:informal1,thm:informal2}). The simpler direction is showing that the
existence of one-way functions implies the existence of adversarial resilient
Bloom filters. Actually, we show that any Bloom filter can be efficiently
transformed to be adversarial resilient with essentially the same amount of
memory. The idea is simple and works in general for other data structures as
well:
apply a pseudo-random permutation of the input and then send it to the
original Bloom filter. The point is that an adversary has almost no advantage
in choosing the inputs adaptively, as they are all randomized by the
permutation, while the correctness properties remain under the permutation.

The other direction is more challenging. We show that if one-way functions do
not exist then any non-trivial Bloom filter can be ``attacked'' by an efficient 
adversary. That is, the adversary performs a sequence of queries and then 
outputs an element $x^*$ (that was not queried before) which is a false 
positive with high probability. We give two proofs: One for the case where the 
Bloom filter has a steady representation and one for an unsteady 
representation.

The main idea is that although we are given only oracle access to the Bloom 
filter, we are able to construct an (approximate) simulation of it. We 
use techniques from machine learning to (efficiently) `learn' the
internal memory of the Bloom filter, and construct the simulation. The 
learning task for steady and unsteady Bloom filters is quite different and 
each 
yield a simulation with different guarantees. Then we show how to exploit each 
simulation to find false positives without querying the real Bloom 
filter.

In the steady case, we state the learning
process as a `PAC learning' \cite{Valiant84} problem. We use what's known as
`Occam's Razor'
which states that any hypothesis consistent on a large enough random training
set will have a small error. Finally, we show that since we assume that one-way functions do
not exist, we are able to find a consistent hypothesis in polynomial time.
Since the error is small, the set of false positive elements
defined by the real Bloom filter is approximately the same set of false
positive elements defined by the simulator.

Handling Bloom filters with an unsteady representation is much more complex.
Recall that such Bloom filters are allowed to randomly change their internal
representation after each query. In this case, we are trying to learn a
distribution that might change after each sample. We describe two examples of
Bloom filters with unsteady representations which seem to capture the main
difficulties of the unsteady case.

The first example considers any ordinary Bloom filter with error rate
$\eps/2$, where we modify the query algorithm to first answer `1' with
probability $\eps/2$ and otherwise continue with its original behavior. The
resulting Bloom filter has an error rate of $\eps$. However, its behavior is
tricky: When observing its responses, elements can alternate between being
false positive and negatives, which makes the learning task much harder.

The second example consists of two ordinary Bloom filters with error rate
$\eps$, both initialized with the set $S$. At the beginning, only the first
Bloom filter is used, and after a number of queries (which may be chosen
randomly) only the second one is used. Thus, when switching to the second
Bloom filter the set of false positives changes completely. Notice that while
first Bloom filter was used exclusively, no information was leaked about the
second. This example proves that any algorithm trying to `learn' the memory of
the Bloom filter cannot perform a fixed number of samples (as does our
learning algorithm for the steady representation case).

To handle these examples, we apply the framework of adaptively changing
distributions (ACDs) presented by Naor and Rothblum \cite{NaorR06}, which 
models the task of learning
distributions that can adaptively change after each sample was studied. Their main result
is that if one-way functions do not exist then there exists an efficient
learning algorithm that can approximate the next activation of the ACD, that
is, produce a distribution that is statistically close to the distribution of
the next activation of the ACD. We show how to facilitate (a slightly
modified version of) this algorithm to learn the unsteady Bloom filter and
construct a simulation. One of the main difficulties is that since we get only 
a statistical distance guarantee, a false positive for the simulation 
need not be a false positive for the real Bloom filter. Nevertheless, we show 
how to estimate whether an element is a false positive in the real Bloom 
filter.

\subsection{Computationally Unbounded Adversaries}

In~\Cref{thm:informal3} we construct a Bloom Filter that is resilient against
any unbounded adversary for a given number ($t$) of queries. One immediate 
solution would be to imitate the construction of the computationally bounded 
case while replacing the pseudo-random permutation with a $k=(t+n)$-wise 
independent hash function. Then, any set of $t$ queries along with the $n$ 
elements of the set would behave as truly random under the hash function. The 
problem with this approach is that the representation of the hash function is 
too large: It is $O(k\log{|U|})$ which is more than the number of 
bits needed for a precise representation of the set $S$. Turning to almost 
$k$-wise independence does not help either. First, the memory will still be 
too large (it can be reduced to $O(n\log{n}\leps + 
t\log{n}\leps)$ 
bits) and second, we get that only sets chosen in advance will act as random, 
where the 
point of an adversarial resilient Bloom filter is to handle adaptively chosen 
sets.

Carter \etalcite{CarterFGMW78} presented a general transformation from any 
exact 
dictionary to a Bloom filter. The idea was
simple: storing $x$ in the Bloom filter translates to storing $g(x)$ in a
dictionary for some (universal) hash function $g:U \rightarrow V$, where 
$|V|=\frac{n}{\eps}$. The choice of the
hash function and underlying dictionary are important as they determine the 
performance and memory size of the Bloom filter. Notice that, at this point 
replacing $g$ with a $k=(t+n)$-wise independent hash function (or an almost 
$k$-independent hash function) yields the same problems discussed above. 
Nevertheless, this is the starting point for the construction, and we show how 
to overcome these issues. 
Specifically, we combine two main ingredients: Cuckoo hashing and a highly 
independent hash function tailored for this construction.

For the underlying dictionary in the transformation, we use the Cuckoo hashing 
construction \cite{PaghR04,Pagh08}. Using cuckoo hashing as
the underlying dictionary was already shown to yield good constructions for
Bloom filters by Pagh~\etalcite{PaghPR05} and 
Arbitman~\etalcite{ArbitmanNS10}. Among the many advantages of Cuckoo 
hashing (\eg succinct memory representation, constant lookup time) is the 
simplicity of its structure. It consists of two tables $T_1$ and $T_2$ and two 
hash functions $h_1$ and $h_2$ and each element $x$ in the Cuckoo dictionary 
resides in either $T_1[h_1(x)]$ or $T_2[h_2(x)]$. However, we use this 
structure a bit differently. Instead of storing $g(x)$ in the dictionary 
directly (as the reduction of Carter et al.\ suggests) which would resolve to 
storing $g(x)$ at either $T_1[h_1(g(x))]$ or $T_2[h_2(g(x))]$ we store $g(x)$ 
at either $T_1[h_1(x)]$ or $T_2[h_2(x)]$. That is, we use the full description 
of $x$ to decide where $x$ is stored but eventually store only a hash of $x$ 
(namely, $g(x)$). Since each element is compared only with two cells, we can 
reduce the size of 
$V$ to $O\paren{\frac{1}{\eps}}$ (instead of $\frac{n}{\eps}$).

To initialize the hash function $g$, instead of using a universal hash 
function we use a very high independence function (which in turn is also 
constructed based on cuckoo hashing) based on the work 
of Pagh and Pagh~\cite{PaghP08} and Dietzfelbinger and
Woelfel~\cite{DietzfelbingerW03}. They showed how to construct a family $G$ of 
hash functions so that on any given set of $k$ inputs it behaves like a truly 
random function with high probability. Furthermore, a function in $G$ can be
evaluated in constant time (in the RAM model), and its description can be
stored using roughly $O(k\log{|V|})$ bits (where $V$ is the range of the 
function).

Note that the guarantee of the function acting randomly holds only for
sets $S$ of size $k$ that are \emph{chosen in advance}. In our case, the set
is not chosen in advance but rather chosen adaptively and adversarially. 
However,
Berman \etalcite{BermanHKN13} showed that the construction of
Pagh and Pagh still works {\em even when the set of queries is chosen
adaptively}.

At this point, one solution would be to use the family of functions $G$ 
setting $k=t+n$, with the analysis of Berman et al.\ as the hash function $g$ 
and the structure of the Cuckoo hashing dictionary. To get an error of $\eps$, 
we set $|V| = O\paren{\leps}$ and get an adversarial resilient Bloom filter 
that is 
resilient for $t$ queries and uses $O\paren{n\leps + t\leps}$ bits of memory. 
However, our goal is to get a memory size of $O\paren{n\leps + t}$.

To reduce the memory of the Bloom filter even further, we use the family $G$ a 
bit differently. Let $\ell = O\paren{\leps}$, and set $k=O\paren{t/\ell}$. We 
define the 
function $g$ to be a concatenation of $\ell$
independent instances $g_i$ of functions from $G$, each outputting a single
bit ($V=\B$). Using the analysis of Berman et al.\ we get that
each of them behaves like a truly random function for any sequence of $k$
adaptively chosen elements. Consider an adversary performing $t$ queries. To 
see how this composition of 
hash functions helps reduce the independence needed, consider the 
comparisons performed in a query between $g(x)$ and some value $y$ being 
performed bit by bit. Only if the first pair of bits are equal we continue to 
compare the next pair. The next query continues from
the last pair compared, in a cyclic order. For any set of $k$ elements, the 
probability of the two bits to be equal is $1/2$. Thus, with high probability, 
only a constant number of bits will be compared during a single query. That 
is, in each query only a 
constant number of function $g_i$ will be involved and ``pay'' in their 
independence, where the rest remain untouched. Altogether, we get that 
although there are $t$ queries performed, we have $\ell$ 
different functions and each function $g_i$ is involved in at most 
$O(t/\ell) = k$ queries (with high probability). Thus,
the view of each function remains random on these elements. This results 
in an adversarial resilient Bloom filter that is resilient for $t$ queries and 
uses only $O(n\leps + k\leps)=O(n\leps + t)$ bits of memory.

\section{Adversarial Resilient Bloom Filters and One-Way Functions}
\label{sec:essential}
In this section, we show that adversarial resilient Bloom filters are
(existentially) equivalent to one-way functions (see Definition 
\ref{def:onewayfunctions}). We 
begin by showing that
if one-way functions do not exist, then any Bloom filter can be ``attacked''
by an efficient algorithm in a strong sense: 
\begin{theorem}\label{thm:no-owf}
Let $\BB=(\BB_1,\BB_2)$ be any non-trivial Bloom filter (possibly with unsteady 
representation) of $n$ elements that
uses $m$ bits of memory and let $\eps_0$ be the minimal error of 
$\BB$.\footnote{The definition of non-trivial is according to 
\Cref{def:non-trivial}. The definition of stead and unsteady representations 
see \Cref{def:steady,def:unsteady} respectively. The minimal error of a Bloom 
filter is defined in \Cref{def:minimal_error}.} If
one-way functions do not exist, then for any constant $\eps < 1$, $\BB$ is not
 $(n,t,\eps)$-adversarial resilient for $t=O\paren{m/\eps_0^2}$.
\end{theorem}

We give two different proofs; The first is self-contained (in particular, we do 
use the Impagliazzo-Luby~\cite{ImpagliazzoL89} 
technique of finding a random inverse), but, deals only with Bloom filters with 
steady representations.
The second handles Bloom filters with unsteady representations, and uses the
framework of adaptively changing distributions  of~\cite{NaorR06}.

\subsection{A Proof for Bloom Filters with Steady Representations.}\label{prf:steady}
\paragraph{Overview:} We prove~\Cref{thm:no-owf} for Bloom filters
with steady representation (see Definition~\ref{def:steady}). Actually, for
the steady case the theorem holds even for $t=O\paren{m/\eps_0}$. In both 
cases, the adversary $A_1$ chooses a uniformly random set $S$. Thus, we focus 
on describing the adversary $A_2$.

Assume that there are no
one-way functions. We want to construct an adversary that can attack the
Bloom filter. We define a function
$f$ to be a function that gets a set $S$, random bits $r$, and
elements $x_1,\ldots,x_t$, computes $\M=\BB_1(S;r)$ and outputs 
$x_1,\ldots,x_t$ along with their
evaluation on $\BB_2(\M,\cdot)$ (i.e.\ for each element $x_i$ the value
$\BB_2(\M,x_i)$).
Since $f$ is not one-way, there is an efficient
algorithm that can invert it with high probability\footnote{The algorithm can
invert the function for infinitely many input sizes. Thus, the adversary we
construct will succeed in its attack on the same (infinitely many) input
sizes.}. That is, the algorithm is
given a random set of elements labeled with the information whether they are 
(false) positives or
not and it outputs a set $S'$ and bits $r'$. For $\M'=\BB_1(S';r')$
the function $\BB_2(\M',\cdot)$ is consistent with $\BB_2(\M,\cdot)$ for all
the elements $x_1,\ldots,x_t$. For a large enough set of
queries we show that $\BB_2(\M',\cdot)$ is actually a good approximation of
$\BB_2(\M,\cdot)$ as a function from $U$ to $\B$.

We use $\BB_2(\M',\cdot)$ to find an
input $x^* \not \in S$ such that
$\BB_2(\M',x^*)=1$ and show that $\BB_2(\M,x^*)=1$ as well (with high probability). This 
contradicts $\BB$ being
adversarial-resilient and proves that $f$ is a (weak) one-way function (see 
Definition \ref{def:weakOWF}).

\begin{proof}[Proof of \Cref{thm:no-owf}(for the steady case).]
Let $\BB=(\BB_1,\BB_2)$ be an adversarial-resilient Bloom filter (see 
Definition \ref{def:adversarialsteady}) that uses $m$
bits of memory, initialized with a random set $S$ of size $n$, and let
$\M=\BB_1(S)$ be its representation. Assume that one-way functions do not
exist. Our goal is to construct an algorithm that, given access to
$\BB_2(M,\cdot)$, finds an element $x^* \notin S$ that is
a false positive with probability greater than $\eps$. For the simplicity of 
presentation, we assume $\eps \le 2/3$ (for other values of $\eps$ the same 
proof works while adjusting the constants appropriately). We need to show an 
attack on infinitely many $n$'s, and thus throughout the proof we assume that 
$n$ and $m$ are large enough.
We describe the function $f$ (which we intend to invert).

\paragraph{The Function $f$.}	
The function $f$ takes $N=\log{u \choose n} + r(n) + t\log{u}$ bits as inputs
where $r(n)$ is the number of random coins that $\BB_1$ uses ($r(n)$ is
polynomial since $\BB_1$ runs in polynomial time) and $t =
\frac{200m}{\eps_0}$.
The function $f$ uses the
first $\log{u \choose n}$ bits to sample a set $S$ of size $n$, the next
$r(n)$ bits (denoted by $r$) are used to run $\BB_1$ on $S$ and get
$\M^S_r=\BB_1(S;r)$. The last $t\log{u}$ bits are interpreted as $t$ elements
of
$U$ denoted $x_1, \ldots, x_t$. The output of $f$ is a sequence of these
elements along with their evaluation by $\BB_2(\M^S_r,\cdot)$. Formally, the
definition of $f$ is:
\begin{align*}
	f(S,r,x_1, \ldots ,x_t) = x_1, \ldots ,x_t,\BB_2(\M^S_r,x_1), \ldots
	,\BB_2(\M^S_r,x_t)
\end{align*}
where $\M^S_r=\BB_1(S;r)$.

It is easy to see that $f$ is polynomial-time computable. Moreover, we can obtain the output of $f$ on a uniform input by sampling
$x_1,\ldots,x_t$ and querying the oracle $\BB_2(\M,\cdot)$ on these elements.
As shown before
(see \cite{Goldreich2001} Section 2.3), if one-way functions do not exist then
weak one-way functions (see Definition \ref{def:weakOWF}) also do not exist. 
Thus, we can assume that $f$ is
not weakly one-way. In particular, we know that there exists a polynomial-time 
algorithm $A$
that inverts $f$ with probability at least $1-1/100$. That is:
\begin{align*}
	\Pr[f(A(f(S,r,x_1, \ldots ,x_t))) = f(S,r,x_1, \ldots ,x_t)]	\ge
	1-1/100.
\end{align*}
Using $A$ we construct an algorithm $\Predict$ that will find a false positive
$x^*$ using $t$ queries with probability $2/3$. The description of the
algorithm is given in \Cref{fig:predict_algorithm}.

\begin{figure}[!h]
	\begin{boxedminipage}{\textwidth}
		\small \medskip \noindent
		
		\vspace{3mm} \textbf{The Algorithm \Predict}
		
		\vspace{3mm} \emph{Given}: Oracle access to the query algorithm
		$\BB_2(\M,\cdot)$. The set $S$.
		
		\vspace{3mm} \emph{Input}: $1^\lambda,n,m,u$.
		
		\begin{enumerate}
			\item For $i \in [t]$ sample $x_i \in U$ uniformly at random and
			query $y_i=\BB_2(\M,x_i)$.
			\item Run $A$ (the inverter of $f$) on $(x_1, \ldots
			,x_t,y_1,\ldots,y_t)$ and $1^\lambda$ to get an inverse 
			$(S',r',x_1, 
			\ldots ,x_t)$.
			\item Compute $\M'=\BB_1(S';r')$.
			\item Do $k=\frac{200}{\eps_0}$ times:
			\begin{enumerate}
				\item Sample $x^* \in U \setminus \{x_1,\dots,x_t\} \cup  
				S$ 
				uniformly 
				at random.
				\item If $\BB_2(\M',x^*)=1$ output $x^*$ and HALT.
			\end{enumerate}
			\item Output an arbitrary $x^* \in U$.
		\end{enumerate}
		
		\vspace{1mm}
		
	\end{boxedminipage}
	
	\caption{The description of the algorithm $\Predict$.}
	\label{fig:predict_algorithm}
\end{figure}
We need to show that the success probability of $\Predict$ is more than 
$2/3$. That is, if $x^*$ is the output of the $\Predict$ algorithm then we
want to show that
$\Pr[\BB_2(\M,x^*)=1] \ge 2/3$. Our first step is showing that if $A$
successfully inverts $f$ then with high probability the resulting $\M'$
defines a function
that agrees with $\BB_2(\M,\cdot)$ on almost all points. For any
representations $\M,\M'$ define their error by:
\begin{align*}
\err(\M,\M') :=\Pr_{x \in U}[\BB_2(\M,x) \ne \BB_2(\M',x)],
\end{align*}
where $x$ is chosen uniformly at random from $U$.
Using this notation we prove the following claim which is very similar to what
is known as ``Occam's Razor'' in learning theory \cite{BlumerEHW89}.
\begin{claim}
Let $t = \frac{1000m}{\eps_0}$. Then, for any representation $\M$, over the 
random choices of $x_1,\ldots,x_t$,
the probability that there exists a representation $\M'$ that is consistent
with $\M$ (\ie that for $i \in [t]$, $\BB_2(\M,x_i)=\BB_2(\M',x_i)$) and
$\err(\M,\M') > \frac{\eps_0}{100}$ is at most $1/100$.
\end{claim}
\begin{proof}
Fix $\M$ and consider any $\M'$ such that $err(\M,\M') > \eps_0/100$. We want
to bound from above the probability over the choice of $x_i$'s that $\M'$ is
consistent with $\M$ on $x_1, \ldots, x_t$. From the independence of the
choice of the $x_i$'s we get that
\begin{align*}
\Pr_{x_1, \ldots, x_t}\left[\forall i \in [t] \colon
\BB_2(\M,x_i)=\BB_2(\M',x_i)\right] \le \paren{1-\frac{\eps_0}{100}}^t.
\end{align*}
Since the data structure uses $m$ bits of memory, there are at most $2^m$
possible representations and at most $2^m$ candidates for $\M'$.
Taking a union bound over the all candidates and for $t = \frac{800m}{\eps_0}$
we get that the probability that there exists such a $\M'$ is:
\begin{align*}
\Pr_{x_1,\dots,x_t}\left[\exists \M' \colon \forall i \in [t],
\BB_2(\M,x_i)=\BB_2(\M',x_i)\right]
\le 2^m\paren{1-\frac{\eps_0}{100}}^t \le
2^m \cdot e^{-10m} \le \frac{1}{100}.
\end{align*}
By the definition of $f$, if $A$ successfully inverts $f$ then it must output 
$S'$ and $r'$ such that $\M'=\BB_1(S';r')$ is a representation that is 
consistent with $\M$ on all samples $x_1,
\ldots, x_t$. Thus, assuming $A$
inverts successfully we get that the probability that $\err(\M,\M') >
\eps_0/100$ is at most $1/100$.
\end{proof}

Given that $\err(\M,\M') \le \eps_0/100$, we want to show that with high
probability step 4 will halt (on step 4.b). Define
$\mu(\M)=\Pr_{x \in U}[\BB_2(\M,x)=1]$ to be the fraction of positives (false 
and
true) of $\M$. The number of false positives might depend on $S$. For
instance, a Bloom filter might store the set $S$ precisely if $S$ is some
special set fixed in advance, and then $\mu(M)=n$. However, we show that for 
most sets the fraction of positives must be approximately $\eps_0$.
\begin{claim}\label{clm:many_positives}
For any Bloom filter with minimal error $\eps_0$ it holds that:
\begin{align*}
\Pr_S \left [\exists r:\mu\paren{\M^S_r} \le \frac{\eps_0}{8} \right] \le 
2^{-n}
\end{align*}
where the probability is taken over a uniform choice of a set $S$ of size $n$
from the universe $U$.
\end{claim}
\begin{proof}
Let $\BAD$ be the set of all sets $S$ such that there exists an $r$ such
that $\mu(\M^S_r) \le \frac{\eps_0}{8}$. Since the number of sets $S$ is
${u \choose n}$, we need to show that $|\BAD| \le 2^{-n}{u \choose n}$. We
show this using an encoding argument for $S$. Given $S \in \BAD$ there is
an $r$ such that $\mu(\M^S_r) \le \frac{\eps_0}{8}$. Let $\widehat{S}$ be the
set of all elements $x$ such that $\BB_2(\M^S_r,x)=1$. Then, $|\widehat{S}| \le
\frac{\eps_0 u}{8}$, and we can encode the set $S$ relative to $\widehat{S}$
using the
representation $\M^S_r$: Encode $\M^S_r$ and then specify $S$ from all
subsets of $\widehat{S}$ of size $n$. This encoding must be more than
$\log{|\BAD|}$ bits and hence we get the bound:

\begin{align*}
	\log{|\BAD|} & \le m + \log{{{\eps_0 u /8} \choose n}} 
	 \le m + n\logp{\eps_0 u /8} - n\log{n} + 2n 
	 \le -n + \log{{u \choose n}},
\end{align*}

where the second inequality follows from \Cref{clm:logbinom}.
\end{proof}

Assuming that $\M'$ is an approximation of $\M$, it follows that $\mu(\M')
\approx \mu(\M)$, and therefore in step 4 with high probability we will find
an $x^*$ such that $\BB_2(\M,x^*)=1$. Namely, we get the following claim.
\begin{claim}\label{clm1}
Assume that $\err(\M,\M') \le \eps_0/100$ and that $\mu(\M) \ge \eps_0/8$.
Then, with probability at least $1-1/100$ the algorithm $\Predict$ will halt
on step 4, where the probability is taken over the internal randomness of
$\Predict$ and over the random choices of $x_1,\ldots,x_t$.
\end{claim}
\begin{proof}
Recall that $\BB$ is non-trivial and by definition, 
we have that $\eps_0 > am/u$ (for any constant $a$) and that there exists a 
constant $c \ge 1$ such 
that $\eps > 1/n^c$. Since $\err(\M',\M) \le \eps_0/100$ and $\mu(\M) \ge 
\eps_0/8$ we get that
$$\mu(\M') \ge \frac{\eps_0}{8} - \frac{\eps_0}{100} > \frac{\eps_0}{10}.$$
Let $\widehat{S'}=\{x:\BB_2(\M',x)=1\}$ be the set of positives (false and 
true) relative to $M'$.
Let $\Xc = \{x_1 \till x_t\}$ be a multiset of the $t$ elements sampled by the 
algorithm. 
The probability of each element to be sampled from $\widehat{S'}$ is 
$\mu(\M')$, and thus the expectation is:
$$\Eparen{|\widehat{S'} \cap \Xc|} = t\cdot \mu(\M') \ge \frac{1000m}{\eps_0} 
\cdot \frac{\eps_0}{10} = 100m.$$
By a Chernoff bound we get that:
\begin{align*}
\Pr\left[|\widehat{S'} \cap \Xc| < 200m \right] > 1-e^{-\Omega\paren{m}}.
\end{align*}
Thus, choosing a large enough constant $a$ we get that
\begin{align*}
|\widehat{S'} \setminus (\Xc \cup S)| \ge |\widehat{S'}| - |\widehat{S'} 
\cap \Xc|  - |S| \ge u \cdot \mu(M') - 200m - n \ge \frac{u\eps_0}{10} - 200m - 
n.
\end{align*}
Thus, we can bound the probability of sampling a successful $x^*$
\begin{align*}
\Pr_{x^* \in U \setminus \{x_1,\dots,x_t\} \cup S}[x^* \in \widehat{S'}] \ge 
\frac{|\widehat{S'} \setminus (\Xc \cup S)|}{u} \ge \frac{\eps_0}{10} - 
\frac{200m}{u} - n/u \ge \frac{\eps_0}{10} - \frac{\eps_0}{20} = 
\frac{\eps_0}{20},
\end{align*}
where the last inequality holds since $u > \frac{am}{\eps_0}$ for any constant 
$a>0$.
That is, the probability of failing is at most $1-\eps_0/20$ and thus the 
probability of failing in all $k=\frac{100}{\eps_0}$ attempts is at most
$$
\paren{1-\frac{\eps_0}{20}}^{k} =
\paren{1-\frac{\eps_0}{20}}^{\frac{200}{\eps_0}}  <
1/100,
$$
and the claim follows.
\end{proof}
Assume that we found a random element $x^*$ that was never queried before such
that $\BB_2(\M',x^*)=1$. Since $\err(\M',\M) \le \eps/100$ we have that
\begin{align*}
	\Pr_x[\BB_2(\M,x^*) = 0 ~\mid~ \BB_2(\M',x^*)=1] \le 1/100.
\end{align*}
Altogether, taking a union bound on all the failure points we get that the
probability of $\Predict$ to fail is at most $4/100 < 1/3$ as required.
\end{proof}

\subsection{Handling Unsteady Bloom Filters}
In the previous section, we have given a proof for \Cref{thm:no-owf} for any 
{\em steady} Bloom filter. In this section, we describe the proof of the 
theorem that handles Bloom filters with an {\em unsteady} representation as 
well. A Bloom filter
with an unsteady representation (see Definition \ref{def:unsteady}) has 
a \emph{randomized} query algorithm and may
change the underlying representation after each query. We want to show
that if one-way functions do
not exist then we can construct an adversary, $\PredictUnsteady$, that 
`attacks' this Bloom filter.

\paragraph{Hard-core Positives.}
Let $\BB=(\BB_1,\BB_2)$ be an $(n,\eps)$-Bloom filter with an unsteady 
representation that uses $m$ bits of memory (see Definition 
\ref{def:unsteady}). Let $\M$ and $\M'$ be two representations of a set $S$ 
generated by $\BB_1$. In the previous proof in \Cref{prf:steady}, given a 
representation $\M$ we considered $\BB_2(M,\cdot)$ as a boolean function. We 
defined the function $\mu(\M)$ to measure the number of positives in 
$\BB_2(M,\cdot)$ and we defined the error between two representations 
$\err(\M,\M')$ to measure the fraction of inputs that the two boolean 
functions agree on. These definitions make sense only when $\BB_2$ is 
deterministic and does not change the representation. However, in the case of 
Bloom filters with unsteady representations, we need to modify the definitions 
to have new meanings.

Given a representation $\M$ consider the query interface $Q(\cdot)$ 
initialized with $\M$. For an
element $x$, the probability of $x$ being a false positive is
$\Pr[Q(x)=1]=\Pr[\BB_2(M,x)=1]$. Recall that after querying
$Q(\cdot)$, the interface updates its representation and the probability of $x$
being a false positive might change (it could be higher or lower). We say that
$x$ is a `hard-core positive' if after any arbitrary sequence of queries we
have that $\Pr[Q(x)=1]=1$. That is, the query interface will always respond
with a `1' on $x$ even after any sequence of queries. Then, we define
$\mu(\M)$ to be the set of hard-core positive elements in $U$. Note that
over the time, the size of $\mu(\M)$ might grow, but by definition it can never 
become smaller. We observe that what Claim \ref{clm:many_positives} actually 
proves is
that for almost all sets $S$ the number of \emph{hard-core} positives is large.

\paragraph{The Distribution $D_{\M}$.}
As we can no longer talk about the \emph{function} $\BB_2(\M,\cdot)$ we turn
to
talk about \emph{distributions}. For any representation $\M$, define the
distribution $D_{\M}$: Sample $k$ elements at random $x_1,\ldots,x_{k}$ ($k$ 
will be determined later), and output
$(x_1,\ldots,x_{k},Q(x_1),\dots, Q(x_{k}))$. Note that the underlying
representation $\M$ changes after each query. The precise algorithm of
$D_{\M}$ is given by:

\noindent\uline{$D_M$:}
\begin{enumerate}
	\item Sample $x_1,\ldots,x_{k} \in U$ uniformly at random.
	\item For $i=1 \till k$: compute $y_i=Q(x_i)$.
	\item Output $(x_1,\ldots,x_{k},y_1,\dots, y_{k})$.
\end{enumerate}

Let $\M_0$ be a representation of a random set $S$ generated by $\BB_1$, and 
let $\eps_0$ be the
minimal error of $\BB$. Let $D_{\M_0}$ be the distribution described above 
using $M_0$. Assume that one-way functions do not exist. Our goal
is to construct an algorithm $\PredictUnsteady$ that will
`attack' $\BB$, that is, it is given 
access to $Q(\cdot)$ initialized with $\M_0$
($\M_0$ is secret and not known to $\PredictUnsteady$) and it must find a 
non-set element $x^*$ such that $\Pr[Q(x^*)=1] \ge 2/3$.

Consider the distribution $D_{\M_0}$, and notice that given access to
$Q(\cdot)$ we can perform a single sample from $D_{\M_0}$: sample $k$ random 
elements and query $Q(\cdot)$ on these elements. Let $M_1$ be
the random variable of the resulting representation after the sample. Now, we
can sample from the distribution $D_{\M_1}$, and then $D_{\M_2}$ and so
on. We begin by describing a simplified version of the proof where we assume 
that $\M_0$ 
is known to the adversary. This version seems to capture the main ideas. 
Then, in \Cref{DynamicFullProof} we show how to eliminate this assumption and 
get a full proof.

\paragraph{Attacking when $\M_0$ is known.}
Suppose that after activating $D_{\M_0}$ for $r$ rounds we are given the 
initial state $\M_0$ (of course, in the actual execution $\M_0$ is secret and 
later we show how to overcome this assumption). Let $p_1, \ldots, p_r$ be the 
outputs of the rounds (that is, $p_i=(x_1 \till x_k,y_1 \till y_k)$). For a 
specific output $p_i$ we say that $x_j$ was labeled `1' if $y_j=1$.

We define a few variants of the distribution $D_{\M_0}$. These variants all 
have same output format (\ie they output $(x_1 \till x_k,y_1 \till y_k)$), and 
differ only by their given probabilities. 
Denote by $D_{\M_0}(p_0 \till p_r)$ the distribution over the $\ith{(r+1)}$
activation of $D_{\M_0}$ conditioned on the first $r$ activations resulting in
the states $p_0 \till p_r$. Computational issues aside, the
distribution $D_{\M_0}(p_0 \till p_r)$ can
be sampled by enumerating all random strings such that when used to run
$D_{\M_0}$ yield the output $p_0,\ldots,p_r$, sampling one of them at random, 
using it to run $D_{\M_0}$ and outputting $p_{r+1}$.

Moreover, define $D_{\M_0}(p_0 \till p_r;x_1 \till x_k)$ to be the same 
distribution as $D_{\M_0}(p_0 \till p_r)$, however, we further conditioned that 
its final output $p_{r+1}$ satisfy $p_{r+1}=x_1,\ldots,x_k,y_1,\ldots,y_k$. 
That is, we condition the $x's$ and the distribution is only 
over the $y_i$'s. Finally, we define $D(p_0 \till
p_r)$ to be the same distribution as $D_{\M_0}(p_0 \till p_r)$ only where the
representation $\M_0$ is also chosen at random (according to $\BB_1(S)$).

We define an (inefficient) adversary $\PredictUnsteady$ in \Cref{fig:attack2} 
that
(given $\M_0$) can attack the Bloom filter, that is, find an element $x^*$
that was not queried before and is a false positive with high probability (the 
final adversary will be efficient).

\begin{figure}[!h]
	\begin{boxedminipage}{\textwidth}
		\small \medskip \noindent
		
		\vspace{3mm} \textbf{The Algorithm \PredictUnsteady}
		
		\vspace{3mm} \emph{Given}: The representation $\M_0$, states 
		$p_1,\ldots,p_r$ and oracle access to $Q(\cdot)$.
		
	\begin{enumerate}
		\item Sample $x_1 \till x_{k} \in U$ at random.
		\item For $i \in [\ell]$ sample $D_{\M_0}(p_0 \till p_r;x_1 \till
		x_k)$ to get $y_{i1} \till y_{ik}$.
		\item If there exists an index $j \in [k]$ such that for all $i \in
		[\ell]$ it holds that $y_{ij}=1$:
		\begin{enumerate}
			\item Set $x^*=x_j$.
			\item Query $Q(x_1) \till Q(x_{j-1})$.
			\item Output $x^*$ and halt.
		\end{enumerate}
		\item Otherwise set $x^*$ to be an arbitrary element in $U$.
	\end{enumerate}
		
		\vspace{1mm}
		
	\end{boxedminipage}
	
	\caption{The description of the algorithm $\PredictUnsteady$.}
	\label{fig:attack2}
\end{figure}
Set $k=160/\eps_0$ and $\ell=100k$. Then we get the following claims. First, we 
show that there will be a common element $x_j$ (that is, the condition in line 
3 will hold).
\begin{claim}\label{clm:1}
With probability $99/100$ there exist a $1 \leq j 
\leq k$ such that for all $i \in
[\ell]$
it holds that $y_{ij}=1$, where the probability is over the random choice of
$S$ and $x_1 \till x_k$.
\end{claim}
\begin{proof}
Let $\M_r$ be the resulting representation of the $\ith{r}$ activation of 
$D_{\M_0}(p_0
\till p_r;x_1 \till x_k)$. We have seen that with probability $1-2^{-n}$ over
the choice of $S$ for any $\M_0$ we have that the set of hard-core positives
satisfy $|\mu(\M_0)| \ge \eps_0/16$. By the definition of the hard-core
positives, the set $\mu(\M_0)$ may only grow after each query. Thus, for each
sample from $D_{\M_0}(p_0 \till p_r;x_1 \till x_k)$ we have that $\mu(\M_0)
\subseteq \mu(\M_r)$. If $x_j \in \mu(\M_0)$ then $x_j \in \mu(\M_r)$ and thus
$y_{ij}=1$ for all $i \in [\ell]$. The probability that all elements $x_1 \till
x_k$ are sampled outside the set $\mu(\M_0)$ is at most $(1-\eps_0/16)^k \le
e^{-10}$ (over the random choices of the elements). All together we get that
probability of choosing a `good' $S$ and a `good' sequence $x_1 \till x_t$ is 
at least $1-2^{-n}
+ e^{-10} \ge 99/100$.
\end{proof}

\begin{claim}\label{clm:2}
Let $\M_r$ be the underlying representation of the interface $Q(\cdot)$ at the
time right {\rm after} sampling $p_0 \till p_r$. Then, with probability at least $98/100$
the algorithm $\PredictUnsteady$ outputs an element $x^*$
such that $Q(x^*)=1$, where the probability is taken over the
randomness of $\PredictUnsteady$, the sampling of $p_0 \till p_r$, and $\BB$.
\end{claim}
\begin{proof}
Consider the distribution $D_{\M_0}(p_0 \till p_r;x_1 \till x_k)$ to work in
two phases: First, a representation $\M$ is sampled conditioned on starting
from $\M_0$ and outputting the states $p_0 \till p_r$ and then we compute
$y_j=\BB_2(\M,x_j)$. Let $\M'_1 \till \M'_{\ell}$ be the representations chosen
during the run of $\PredictUnsteady$. Note that $M_r$ is chosen from the same
distribution that $\M'_1 \till \M'_{\ell}$ are sampled from. Thus, we can
think of $\M_r$ as being picked after the choice of
$x_1 \till x_k$. That is, we sample $\M'_1 \till \M'_{\ell+1}$, and choose one
of them at random to be $\M_r$, and the rest are relabeled as $\M'_1 \till
\M'_{\ell}$. Now, for any $x_j$, the probability that for all $i$, $\M'_i$
will answer `1' on $x_j$ but $\M_r$ will answer `0' on $x_j$ is at most
$1/\ell$. Thus, the probability that there exist any such $x_j$ is at most
$\frac{k}{\ell} = \frac{k}{100k} = 1/100$. Altogether, the probability that 
$A$ finds such an $x_j$
that is always labeled `1' and that $\M_r$ answers `1' on it, is at least
$99/100-1/100=98/100$.
\end{proof}

Finally, we claim that $x^*$ is truly a false positive, that is $x^* \notin S$ 
and that $x^*$ is not one of the previous $t$ queries. By our bound on $t$ we 
get that $n+t = O(m/\eps_0^2)$. Since the Bloom filter is non-trivial, we get 
that $u > am/\eps_0^2$ for any constant $a$. Therefore, the probability that 
$x^*$ is not a new false positive is negligible.

We are left to show how to construct the algorithm $\PredictUnsteady$ so that 
it will run
in polynomial-time and perform the same tasks {\em without} knowing $\M_0$. 
Note that our only use of $\M_0$ was to sample from $D_{\M_0}(p_0 \till 
p_r;x_1 \till x_k)$ without changing it (since it changes after each sample). 
The goal is to observe outputs of this distribution until we have enough 
information about it, such that we can simulate it without changing its state. 
One difficulty (which was discussed in
\Cref{sec:tech}), is that the number of samples $r$ must be chosen
as a function of the samples and cannot be
fixed in advance. Algorithms for such tasks were studied in the framework of 
Naor
and Rothblum~\cite{NaorR06} on adaptively changing distributions.

\subsection{Using ACDs}\label{DynamicFullProof}
We continue the full proof of \Cref{thm:no-owf}. First, we give an overview of
the framework of adaptively changing distributions of Naor and Rothblum
\cite{NaorR06}.

\paragraph{Adaptively Changing Distributions.}
An adaptively changing distribution (ACD) is composed of a pair of
probabilistic algorithms, one for generating ($G$) and for
sampling ($D$). The generation algorithm $G$ receives a public input $x \in
\B^n$ and outputs an initial secret state $s_0 \in \B^{s(n)}$ and a public
state
$p_0$:
\begin{align*}
	G \colon R \to S_p \times S_{init},
\end{align*}
where the set of
possible public states is denoted by $S_p$ and the set of secret states is
denoted by $S_s$.
After running the generation algorithm, we can consecutively activate a
sampling algorithm $D$ to generate samples from the adaptively changing
distribution. In each activation, $D$ receives as its input a pair of secret
and public states, and outputs new secret and public states.  Each new state is 
always a function of the current state and
some randomness, which we assume is taken from a set $R$. The states generated
by an ACD are determined by the function
\begin{align*}
D \colon S_p \times S_s \times R \to S_p \times S_s.
\end{align*}
When $D$ is activated for the first time, it is run on
the initial public and secret states $p_0$ and $s_0$ respectively, generated
by $G$. The public output of the process is a sequence of public states $(x,
p_0, p_1, \ldots)$.

\paragraph{Learning ACDs:}
An algorithm $L$ for learning an ACD $(G,D)$ sees $x$ and $p_0$ (which were
generated,
together with $s_0$, by running $G$), and is then allowed to observe $D$ in
consecutive activations, seeing only the public states $D$ outputs. The
learning algorithm's goal is to output a hypothesis $h$ on the initial secret
state that is functionally equivalent to $s_0$ for the next
activation of $D$. The requirement is that with probability
at least $1 - \delta$ (over the random coins of $G$, $D$, and the learning
algorithm), the distribution of the next public state,
given the past public states and that $h$ was the initial secret output of $G$,
is $\eps$-close to the same distribution with $s_0$ as the initial secret
state (the ``real'' distribution of $D$'s next public state).
Throughout the learning process, the sampling algorithm $D$ is run
consecutively, changing the public (and secret) state. Let $p_i$ and $s_i$ be
the public secret states after $D$'s $\ith{i}$ activation. We refer to the
distribution $D^{s_0}_i(x,p_0,\dots,p_i)$ as the distribution on the public
state
that will be generated by $D$'s next $\ith{(i+1)}$ activation.

After letting $D$ run for (at most) $r$ steps, $L$ should stop and output some
hypothesis $h$ that can be used to generate a distribution that is close to
the distribution of $D$'s next public output. We emphasize that $L$ sees
only $(x, p_0, p_1, \ldots,p_r)$, while the
secret states $(s_0, s_1,\ldots,s_r)$ and the random coins used by $D$ are
kept hidden from it. The number of times $D$ is allowed to run ($r$) is
determined by the
learning algorithm. We say that $L$ is an $(\eps, \delta)$-learning
algorithm for $(G, D)$, that uses $r$ rounds, if when run in a learning
process for $(G, D)$, $L$ always (for any input $x \in \B^n$) halts and
outputs some hypothesis $h$ that specifies a hypothesis distribution $D_h$,
such that with probability $1-\delta$ it holds
that $\Delta(D^{s_0}_{r+1},D_h) \le \eps$, where $\Delta$ is the statistical
distance between the distributions (see Definition 
\ref{def:statisticaldistance}).

We say that an ACD is hard to $(\eps,\delta)$-learn with $r$ samples if no
efficient learning algorithm can $(\eps,\delta)$-learn the ACD using $r$
rounds. The main result regarding ACDs that we use is an equivalence between 
hard-to-learn ACDs
and almost one-way functions (where an almost one-way function is a function that
are hard to invert for an infinite number of sizes, see 
Definition \ref{def:almostOWF}.)

\begin{theorem}[\cite{NaorR06}]\label{thm:acd}
Almost one-way functions exist if and only if there exists an adaptively
changing distribution $(G,D)$ and polynomials $\delta(n)$, $\eps(n)$, such
that it is hard to $(\eps(n), \delta(n))$-learn the ACD $(G,D)$ with
$O\paren{\frac{\log{ |S_{init}|}}{\eps^2(n)\delta^4(n)}}$ samples.
\end{theorem}

The consequence of \Cref{thm:acd} is that if we assume that one-way functions
do not exist, then no ACD is hard to learn.  For concreteness, we get that,
given an ACD there exists an algorithm $L$ that (for infinitely many input
sizes) with probability at least
$1-\delta$ performs at most $O\paren{\log{ |S_{init}|}}$ samples and produces
an hypothesis $h$ on the initial state and a distribution $D_h$ such
that the statistical distance between $D_h$ and the next activation of the ACD
is at most $\eps$. The point is that $D_h$ can be (approximately) sampled in
polynomial-time.

\paragraph{From Bloom Filters to ACDs.}
As one can see, the process of sampling from $D_{\M_0}$ is equivalent to
sampling from an ACD defined by $D_{\M_0}$. The secret states are the
underlying representations of the Bloom filter, and the initial secret state
is $\M_0$. The public states are the outputs of the sampling. Since the Bloom
filter uses at most $m$ bits of memory for each representation we have that
$|S_{init}| \le 2^m$.

Running the algorithm $L$ on the ACD constructed above, will output a
hypothesis $\M_h$ of the initial representation such that the distribution
$D_{\M_h}(p_0 \till p_r)$ is close (in statistical distance) to the
distribution $D_{\M_0}(p_0 \till p_r)$. The algorithm's main
goal is to estimate whether the weight (\ie the probability) of representations 
$\M$ 
according to $D(p_0 \till p_i)$ such that $D(\M,p_0 \till p_i)$ is
close to $D(\M_0;p_0 \till p_i)$ is high, and then samples such a
representation. The main difficulty of their work is showing that if almost
one-way functions do not exist then this estimation of sampling procedures can
be implemented efficiently. We slightly modify the algorithm $L$ to output 
several such
hypotheses instead of only one. The overview of the
modified algorithm is given below (the only difference from the original 
algorithm is the number of output hypotheses):
\begin{enumerate}
\item For $i \leftarrow 1 \ldots r$ do:
	\begin{enumerate}
	\item Estimate whether the weight of representations $\M$ according to
	$D(p_0 \till p_i)$ such that $D_{\M}(p_0 \till p_i)$ is close to
	$D_{\M_0}(p_0 \till p_i)$ is high. Namely, estimate if there is a subset 
	$\mathcal{M} \subset S_{init}$ such that $\Pr_{D(p_0 \till 
	p_i)}[\mathcal{M}]$ is high and for all 
	$M \in \mathcal{M}$ it holds that $\Delta(D_{\M}(p_0 \till p_i), 
	D_{\M_0}(p_0 
	\till p_i))$ is small.
	
	If the weight estimate is high then
	(approximately) sample $h_1 \till h_{\ell} \leftarrow S_{init}$ according 
	to $D(p_0 \till p_i)$,
	output $h_1 \till h_{\ell}$, and terminate.
	\item Activate $D_{\M_i}$ to sample $p_{i+1}$ and proceed to
	round $i+1$.
	\end{enumerate}
\item Output arbitrarily some $\M$ and terminate.
\end{enumerate}

We run the modified algorithm $L$ on the ACD defined by $D_{\M_0}$ with parameters
$\gamma=\frac{1}{100\ell}$ and $\delta=1/100$ (recall that $k=160/\eps_0$ 
and $\ell=100k$). The output is $h_1 \till h_{\ell}$ with the property that 
with probability at least $1-1/100$ for every $i \in [\ell]$ it holds that
$$\Delta\paren{D_{\M_0}(p_0 \till p_r),D_{h_i}(p_0 \till p_r)} \le 1/100.$$
We 
modify the
algorithm $\PredictUnsteady$ such that the $\ith{i}$ sample from $D_{\M}(p_0 
\till p_r;x_1 \till x_k)$ is replaced with a sample from $D_{h_i}(p_0 
\till p_r;x_1 \till
x_k)$. We have shown that the algorithm $\PredictUnsteady$ succeeds given samples from
$D_{\M_0}(p_0 \till p_r;x_1 \till x_k)$. However, now it is given samples
from distributions that are only `close' to $D_{\M_0}(p_0 \till p_r;x_1 \till 
x_k)$.

Consider these two cases where in one we sample from $D_{\M_0}(p_0 \till p_r)$ 
and in the other we sample from $D_{h_i}(p_0 
\till p_r)$. Each one defines a different distribution, denoted $D^1$ and
$D^2$, respectively. The distribution $D^1$ is defined by sampling $x_1 \till
x_k$ and then sampling $\ell$ times from $D_{\M_0}(p_0 \till p_r;x_1 \till 
x_k)$, and
the distribution $D^2$ is defined by sampling $x_1 \till x_k$ and then
sampling from $D_{h_i}(p_0 
\till p_r;x_1 \till x_k)$ for each $i \in [\ell]$.

Since
$\Delta(D_{\M_0}(p_0 \till p_r),D_{h_i}(p_0 
\till p_r) \le \gamma = \frac{1}{100\ell}$
for any $i \in
[\ell]$, by the triangle inequality we get that $\Delta(D^1,D^2) \le \ell\gamma 
=
1/100$. Let $\BAD$ be the event the algorithm $\PredictUnsteady$
does not succeed in finding an appropriate $x^*$. We have shown that under the
first distribution $\Pr_{D^1}[\BAD] \le 2/100$. Moreover, we have shown that with probability $99/100$ we can find a distribution $D^2$ such that $\Delta(D^1,D^2)\le 1/100$. Taking a union bound, we get that the
probability of the event $\BAD$ under the distribution $D^2$ is
$\Pr_{D^2}[\BAD] \le 2/100 + 1/100 + 1/100 \le 1/3$, as required.

The number of rounds performed by $L$ is $O(m/\gamma)$ and each round we
perform $k$ queries. Thus, the total amount of queries is
$O(mk/\gamma)=O(m/\eps_0^2)$. As shown before, the probability that $x^*$ is 
not a new false positive is negligible.

\subsection{A Construction Using Pseudorandom Permutations.}
\label{sec:sufficient}
We have seen that Bloom filters that are adversarial resilient require using
one-way functions. To complete the equivalence, we show that cryptographic tools and in particular pseudorandom
permutations and functions can be used 
to construct adversarial resilient Bloom filters.
Actually, we show that any Bloom filter can be
efficiently  transformed to be adversarial resilient with essentially the same
amount of memory. The idea is simple and can work in
general for other data structures as well: On any input $x$ we compute a
pseudorandom permutation of $x$ and send it to the original Bloom filter. Recall that a pseudorandom permutation (PRP) is a keyed family of permutations that a random function from the family is indistinguishable from a truly random permutation given only oracle access to the function.

\begin{theorem}\label{thm:sufficient}
Let $\BB$ be an $(n,\eps)$-Bloom filter using $m$ bits of memory. If
pseudorandom permutations exist, then there exists a negligible function $\negl(\cdot)$ such that for security parameter $\lambda$
there exists an
$(n,\eps+\negl(\lambda))$-strongly resilient Bloom filter that uses
$m'=m+\lambda$ bits of memory.
\end{theorem}
\begin{proof}
The main idea is to randomize the adversary's queries
by applying a pseudorandom permutation (see Definition \ref{def:pseudorandomP}) on them; then we 
may consider the queries as random and
not as chosen adaptively by the adversary.

Let $\BB$ be an $(n,\eps)$-Bloom filter using $m$ bits of memory. We will
construct a $(n,\eps + \negl(\lambda))$-strongly resilient Bloom filter
$\BB'$ as follows: To initialize $\BB'$ on a set $S$ we first choose a key $K
\in \B^\lambda$ for a pseudo-random permutation, $\PRP$, over $\B^{\log{u}}$. Let
$$S'=\PRP_K(S)= \{\PRP_K(x) : x \in S\}.$$
Then we initialize $\BB$ with
$S'$. For the query algorithm, on input $x$ we output $\BB(\PRP_K(x))$. Notice
that the only additional memory we need is storing the key $K$ of the $\PRP$ which takes $\lambda$ bits. Moreover, the running time of the query algorithm of $\BB'$ is one pseudo-random permutation more than the query time of $\BB$.

The completeness follows immediately from the completeness of $\BB$. If $x \in
S$ then $\BB$ was initialized with $\PRP_K(x)$ and thus when querying on $x$ we 
will query $\BB$ on $\PRP_K(x)$ which will return `1' from the
completeness of $\BB$.

The resilience of the construction follows from a hybrid argument. Let
$A$ be an adversary that queries $\BB'$ on $x_1 \till x_t$ and outputs $x$ where $x \notin \{x_1 \till x_t\}$. Consider the experiment where the $\PRP$ is
replaced with a truly random permutation oracle $R(\cdot)$. Then, since $x$
was not
queried before, we know that $R(x)$ is a truly random element that was not
queried before, and we can
think of it as chosen before the initialization of $\BB$. From the soundness
of $\BB$ we get that the probability of $x$ being a false positive is at most
$\eps$.

We show that $A$ cannot distinguish between the Bloom filter we constructed
and our experiment (using a random permutation) by more than a negligible advantage.
Suppose that there exists a
polynomial $p(\lambda)$ such that $A$ can attack $\BB'$ and find a false
positive with probability $\eps+\frac{1}{p(\lambda)}$. We will show that we
can use $A$ to
construct an algorithm $A_2$ that can distinguish between a random oracle and a 
$\PRP$ with non-negligible probability. Run
$A$ on $\BB'$ where the $\PRP$ is replaced with an oracle that is either
random or pseudo-random. Answer `1' if $A$ successfully finds a false
positive. Then, we have that
\begin{align*}
\abs{\Pr[A_2^R(1^\lambda)] - \Pr[A_2^{\PRP}(\lambda)]} \ge \abs{\eps - \eps +
\frac{1}{p(\lambda)}} = \frac{1}{p(\lambda)}
\end{align*}
which contradicts the indistinguishability of the $\PRP$ family.

\end{proof}

\paragraph{Constructing Pseudorandom Permutations from One-way Functions.}
Pseudorandom permutations have several constructions from one-way functions
which target different domains sizes.
If the universe is large, one can obtain pseudorandom permutations from
pseudorandom functions using the famed Luby-Rackoff
construction from pseudorandom functions~\cite{LubyR88,NaorR99},
which in turn can be based on one-way functions.

For smaller domain sizes ($u=|U|$ quite close to $n$)  or domains that are not
a power of two the problem is a bit more complicated.
There has been much attention given to this issue in recent years and Morris and
Rogaway~\cite{MorrisR14}
presented a construction that is computable in expected $O(\log u)$ 
applications of a pseudorandom function.
Alternatively,
Stefanov and Shi~\cite{StefanovS12} give a construction of small domain  
pseudorandom permutations, where the asymptotic cost is worse 
($\sqrt{u}\log{u}$), however, in practice, it performs very well in common use 
cases (\eg when $u \le 2^{32}$).

The reason we use permutations is to avoid the event of an element $x
\not\in S$ colliding with
the set $S$ on the pseudorandom function.
If one replaces the pseudorandom permutation with a pseudorandom function,
then a term of $n/u$, a bound on the probability of  this collision, must be added to the false
positive error rate;
other than that  the analysis is the same as in the permutation case.
So unless $u$ is close $n$ this additional error might be tolerated,
and it is possible to replace the use of the pseudorandom permutation with a
pseudorandom function.

\begin{remark}[Alternatives]
The transformation suggested does not interfere with the internal operation of
the Bloom filter implementation and might preserve other properties of the 
Bloom filter as well. For example, it is applicable when the set is not given 
in advance but is provided dynamically among the membership queries, and even 
in the case where the size of the set is not known in advance (as 
in~\cite{PaghSW13}).

An alternative approach is to replace the hash functions used in `traditional'
Bloom filter constructions or those in~\cite{PaghPR05,ArbitmanNS10}
with pseudorandom functions. The potential advantage of doing the analysis per
construction is that we may save on the computation. Since many constructions 
use several hash functions, it might be possible to use the result of
a single pseudorandom evaluation to get all the randomness required for all of 
the hash functions at once.

Notice that, in all the above constructions only the pseudorandom 
function (permutation) key must 
remain secret. That is, we get the same security even when the adversary gets 
the entire memory of the Bloom filter except for the PRF (PRP) key.
\end{remark}
\section{Computationally Unbounded Adversary}

In this section, we extend the discussion of adversarial
resilient Bloom filters to ones against computationally \emph{unbounded}
adversaries. First, notice that the attack of \Cref{thm:no-owf} holds in
this case as well, since an unbounded adversary can invert any function (with
probability 1). Formally, we get the following corollary:
\begin{corollary}\label{col}
Let $\BB=(\BB_1,\BB_2)$ be any non-trivial Bloom filter of $n$ elements that
uses $m$ bits of
memory and let $\eps_0$ be the minimal error of $\BB$. Then for any constant 
$\eps < 1$ there exist a $t$ such that $\BB$ is not $(n,t,\eps)$-adversarial 
resilient against unbounded adversaries and $t=O\paren{\frac{m}{\eps_0^2}}$.
\end{corollary}

As we saw, any $(n,\eps)$-Bloom filter must use at least $n\leps$ bits of 
memory. We show how to construct Bloom Filters that are resilient against 
\emph{unbounded} adversaries for $t$ queries while using only 
$O\paren{n\leps + t}$ bits of memory.

\begin{theorem}\label{thm:construction}
For any $n,t \in \N$, and $0 < \eps < 1/2$ there exists an 
$(n,t,\eps)$-resilient
Bloom filter (against unbounded adversaries) that uses $O(n\leps + t)$ bits of 
memory and has linear setup time and $O(1)$ 
worst-case query time.
\end{theorem}

Our construction uses two main ingredients: Cuckoo hashing
and a very high independence hash family $G$. We begin by describing these 
ingredients.

\paragraph{The Hash Function Family  $G$.}
Pagh and Pagh~\cite{PaghP08} and Dietzfelbinger and
Woelfel~\cite{DietzfelbingerW03} (see also 
Aum{\"{u}}ller~\etalcite{AumullerDW14})
showed how to construct a family $G$ of hash functions $g: U \to
V$ so that on any set of $k$ inputs it behaves like a truly random function
with
high probability ($1-1/\poly(k)$). Furthermore, $g$ can be evaluated in
constant time (in the RAM model), and
its description can be stored using $(1+\alpha)k\log{|V|} + O(k)$ bits (where
here $\alpha$ is
an arbitrarily small constant).

Note that the guarantee of $g$ acting as a
random function holds for any set $S$ that is \emph{chosen in advance}.
In our case the set is not chosen in advance but chosen adaptively and
adversarially.
However, Berman \etalcite{BermanHKN13} showed that the same line of constructions, starting with
Pagh and Pagh, actually holds  {\em even when the set of queries is chosen
adaptively}. That is, for any distinguisher that can adaptively choose $k$
inputs, the advantage of distinguishing a function $g \in_R G$ from a truly
random function is polynomially
small\footnote{Any exactly $k$-wise independent function is also good against
$k$ adaptive queries, but this is not necessarily the case for {\em almost}
$k$-wise.}.

Set $\ell=4\leps$. Our function $g$ will be composed of the concatenation of
$\ell$ one bit functions $g_1, g_2, \ldots, g_{\ell}$ where each $g_i$ is
selected independently from a family $G$ where $V=\B$ and $k=2t/\leps$.
For a random $g_i \in_R G$:
\begin{itemize}
\item There is a constant $c$ (which we can choose) so that for any adaptive
distinguisher that issues
a sequence of $k$ adaptive queries on $g_i$ the advantage of distinguishing
between $g_i$ and an exact  $k$-wise independent function $U \to V$ is
bounded by $\frac{1}{k^c}$.
\item $g_i$ can be represented using $O(k)$ bits.
\item $g_i$ can be evaluated in constant time.
\end{itemize}

Thus, the representation of $g$ requires $O(\ell \cdot k) = O(t)$ bits. The 
evaluation of 
$g$ at
a given point $x$ takes  $O(\ell)=O\paren{\leps}$ time.

\paragraph{Cuckoo Hashing.}
Cuckoo hashing is a data structure for dictionaries introduced by Pagh and
Rodler
\cite{PaghR04}. It consists of two tables $T_1$ and $T_2$, each containing $r$
cells where $r$ is slightly larger than $n$ (that is, $r=(1+\alpha)n$ for some
small
constant
$\alpha$) and
two hash functions $h_1,h_2:U \to [r]$. The elements are
stored in the two tables so that an element $x$ resides at either
$T_1[h_1(x)]$ or $T_2[h_2(x)]$. Thus, the lookup procedure consists of one
memory accesses to each table plus computing the hash
functions. This description ignores insertions, where insertions can be 
performed in expected constant time (or alternatively, all $n$ elements can be 
inserted in linear time with high probability).

Our construction of an adversarial resilient Bloom filter is:
\paragraph{Setup.}
The input is a set $S$ of size $n$. Sample a function $g$ by sampling $\ell$
functions $g_i \in_R G$ and
initialize a Cuckoo hashing dictionary $D$ of size $n$ (with $\alpha=0.1$) as
described above. That is, $D$ has two tables $T_1$ and $T_2$ each of size
$1.1n$, two hash
functions $h_1$ and $h_2$, and each element $x$ will reside at either
$T_1[h_1(x)]$ or $T_2[h_2(x)]$. Insert the elements of $S$ into $D$.
Then, go over the two tables $T_1$ and $T_2$ and at
each cell replace each $x$ with $g(x)$. That is, now for each $x \in S$ we
have that $g(x)$ resides at either $T_1[h_1(x)]$ or $T_2[h_2(x)]$. Put
$\bot$
in the empty locations. The final
memory of the Bloom Filter is the memory of $D$ and the representation of $g$.
The dictionary $D$ consists of $O(n)$ cells, each of size $|g(x)|=O(\leps)$
bits and
therefore $D$ and $g$ together can be represented by $O(n\leps +t)$ bits.

\paragraph{Lookup.}
On input $x$, answer `1' if either $T_1[h_1(x)] = g(x)$ or
$T_2[h_2(x)]=g(x)$, and `0' otherwise.

\begin{theorem}
Let $\BB$ be a Bloom filter as constructed above. Then for any constant $0 < 
\eps < 1/2$, $\BB$ is an
$(n,t,\eps)$-resilient Bloom filter
against unbounded adversaries, that uses $m$ bits of memory where
$m=O(n\leps + t)$.
\end{theorem}
\begin{proof}
Let $A$ be any (unbounded) adversary that performs $t$ adaptive queries $x_1
\till x_t$ on $\BB$. The function $g$
constructed above outputs $\ell$ bits. For the analysis, recall that we
constructed
the function $g$ to be composed of $\ell$ independent functions, each $g_i$
with the range $V=\B$.
For the rest of the proof, we denote $g(x)$ to
be the composition of $g_1(x),\ldots, g_{\ell}(x)$.

In the lookup procedure, on each query $x_i$ we compute $g(x_i)$ and compare
it to a single cell in each table. Suppose that the comparison between
$g(x_i)$ and each cell is done bit by bit. That
is, on the first index where they differ we output '0' (\ie `no') and halt (and 
do not continue to the next bit). Only if
all the bits are equal we answer `1' (\ie `yes'). That is, not all functions 
$g_j$
necessarily participate on each query.

Moreover, suppose that for each cell
we {\em mark the last bit} that was compared. Then, on the next query, we 
continue
the comparison from the next bit in a cyclic order (the next bit of the last bit
is the first bit).

Let $Q=\{x_1 \till x_t\}$ be the set of elements that the adversary queries.
For any $j \in [\ell]$ let $Q_j \subset Q$ be the subset of queries that the
function $g_j$ participated in. That is, if $x_i \in Q_j$ then the function
$g_j$ participated in the comparisons of query $x_i$ (in either one of the
tables). For any set
$Q_j$, if $|Q_j| \le k$, then with high probability $g_j$ is $k$-wise
independent on the set $Q_j$ (in the distinguishing sense) and the
distribution of $g_j(Q_i)$ is uniform
in the view of $A$. Thus, we want to prove the following claim.
\begin{claim}
With probability at least $1-\eps/2$, for all $j \in [\ell]$ it holds that
$|Q_j| \le k$ and that no adversary can distinguish between the evaluation of
$g_j$ on $Q_j$ and the evaluation of a truly random function on $Q_j$,
where the probability is taken over the initialization of the Bloom filter.
\end{claim}
\begin{proof}
Any $(n,\eps)$-Bloom filter is resilient to a small number of queries. If $t$ 
is much less than $1/\eps$, then we do not expect the adversary to see any false 
positives, and hence we can consider the queries as chosen in advance. 
Therefore, we assume (without loss 
of generality) that $t > 1/\sqrt{\eps}$. Moreover, we can assume that $t \ge 
n\leps$, since a smaller $t$ does not reduce the memory use.

Each function $g_j$ is $k$-wise independent on any sequence of queries of
length $k$ with
probability at least $1-\frac{1}{k^c}$ (in the distinguishing sense). 
Therefore, the probability that any one of them is not $k$-wise 
independent (via a union bound) is at most $\ell/k^c = 4\leps/k^c \le \eps/4$
(for a large enough constant $c$).

Suppose  that all the functions are $k$-wise independent.
Since the comparisons are performed in a cyclic order for any $j\in[\ell]$ we 
always
have that
$$|Q_j| \le |Q_1| \le |Q_j|+1,$$
and thus it is enough to bound $|Q_1|$. Let $X_i$
be the number of functions that participated on query $x_i$ (in both tables)
and let $X=\sum_{i=1}^{t}X_i$. Since $g_j$ is $k$-wise
independent hash that outputs a single bit, for any $x \ne x'$ we
have that
$$\Pr[g_j(x')=g_j(x)]=\Pr[g_j(x')=0 \wedge g_j(x)=0] + \Pr[g_j(x')=1 \wedge 
g_j(x)=1] =1/4+1/4=1/2.$$
Therefore, during each comparison we
expect two functions to participate. Taking both tables into account, we get 
that $\E[X_i] \le 4$, and $\E[X] \le 4t$.
Using a Chernoff bound we get that
\begin{align*}
\Pr[X \ge 5t] \le e^{-\Omega(t^2)} \le e^{-\Omega((n\leps)^2)} \le \eps/4.
\end{align*}
On the other hand, since $|Q_1| \le |Q_j|+1$ we have that:
\begin{align*}
X = \sum_{j=1}^{\ell}|Q_j| \ge \ell(|Q_1|-1).
\end{align*}
Thus, if conditioned on having $X \le 5t$ we get that:
$$|Q_1| \le \frac{X}{\ell} + 1 \le \frac{5t}{\ell}+1 = \frac{5t}{4\leps}+1 \le 
\frac{2t}{\leps} = k.$$
Together, we get that the probability that for all $j \in [\ell]$
it holds that $|Q_j| \le k$ 
is at least $1-\eps/4-\eps/4=1-\eps/2$, which completes the proof of the claim.
\end{proof}

Assuming that each function $g_i(\cdot)$ is $k$-wise independent, and $|Q_i|
\le k$ we get that for any query $x$ the distribution of $g_j(x)$ is uniform.
Let $w_1,w_2$ be the contents of the cells that will be compared to with $x$.
The probability that $x$ is a false positive is at most the probability that
$g(x)=w_1$ or $g(x)=w_2$. Thus, we get that
\begin{align*}
\Pr[x \text{ is a false positive}] \le 2\Pr\left[\forall j \in [\ell] \colon
g_j(x)=w(j)\right] = 2 \cdot \paren{1/2}^{\ell} = 2 \cdot \paren{1/2}^{4\leps}
\le 2\eps^4 \le \eps/2.
\end{align*}

Going from exact $k$-wise independence to almost $k$-wise independence adds an 
error probability of $\eps/4$. Therefore, the
overall probability that $x$ is a false positive is at most
$\eps/2 + \eps/2=\eps$ and thus our construction is secure against $t$ queries. 
This completes the proof of the Theorem.
\end{proof}

Note that the time to evaluate the Bloom filter is $O(\leps)$. We can turn the
construction into an $O(1)$ evaluation one at the cost of using additional
$O(t\leps)$ bits instead.

\paragraph{Random Queries}
Our construction leaves a gap in the number of queries needed to attack
$(n,\eps)$-Bloom filters that use $m=O(n\leps)$ bits of memory. The attack
uses $O(m/\eps)$ queries while our construction is resilient to at most
$O(n\leps)$ queries. However, notice that the attack uses only random queries.
That is, it performs $t$ random queries to the Bloom filter and then decides
on $x^*$ accordingly. If we assume that the adversary works in this way, we
can actually show that our constructions is resilient to $O(m/\eps)$ queries,
with the same amount of memory.

\begin{lemma}
In the random query model, for any $n$, $\eps>0$ and $t = n/\eps$ there
exists an $(n,t,\eps)$-resilient Bloom filter that uses $O(n\leps)$ bits of
memory.
\end{lemma}
\begin{proof}
We use the same construction as in \Cref{thm:construction} where we set $\ell
= 2\leps$, $V=\B^{\ell}$, $k=n$ and $\alpha=0.1$. The number of bits required
to represent $g$ is $(1+\alpha)k\log{|V|} + O(k) = O(n\leps)$.

The analysis is also similar, however, in this case, we assume
that the comparisons are always done from the leftmost bit to rightmost one.
Let $X$ be a random variable denoting the number of queries among the $t=n/\eps$
random queries that pass the first $2\leps$ comparisons. The idea is to show
that $X$ will be smaller than $n$ with high probability, and thus the rest of
the $2\leps$ functions remain to act as random for the final query.

For any $j$, since the queries are random we have that $\Pr[g_j(x') \ne
g_j(x)]=1/2$. Thus, the probability that a single random query passes the
first $2\leps$ comparisons is $(1/2)^{2\leps} = \eps^2$. Thus $\E[X] \le
n/\eps
\cdot \eps^2=n\eps \le n/2$. Moreover, since the queries are independent this
expectation is concentrated and using a Chernoff bound we get that with
exponentially high probability $X < n$. For $1 \le j \le 2\leps$ we cannot
bound $|Q_j|$ and indeed it might hold that $|Q_j|$ is much larger than $n$.
However, for $2\leps < j \le 4\leps$ with high probability we have that $|Q_j|
< n$, and since each function is $n$-wise independent it acts as a random
function on $Q_j$. Therefore, for the last query $x^*$, the probability that
it passes the last $2\leps$ queries is at most $(1/2)^{2\leps} = \eps^2$.

Taking a union bound on the event that $g$ is $n$-wise independent and that $X
< n$ we get that $x^*$ will be a false positive with probability smaller than
$\eps$.
\end{proof}

\subsection{Open Problems}
An open problem this work suggests is proving tight bounds on the number of 
queries required to `attack' a Bloom filter. Consider the case where 
the memory of the Bloom filter is restricted to be $m=O(n\leps)$. Then, by 
\Cref{col} we have an upper bound of $O\paren{\frac{m}{\eps_0^2}}$ queries 
(actually $O\paren{\frac{m}{\eps_0}}$ for the steady case). In the random 
query model, the construction provided above (\Cref{thm:construction}) gives 
us an almost tight lower bound of $\Omega\paren{\frac{m}{\eps_0}}$ (which is 
tight in the steady case). However, for arbitrary queries, the lower bound is 
only $\Omega(m)= \Omega(n\leps)$.

Another issue regards the dynamic case (where the set is not given in 
advance). Our lower bounds hold for both cases and the construction 
against computationally bounded adversaries hold in the dynamic case as well. 
An open problem is to adjust the construction in the unbounded case to handle 
adaptive inserts while maintaining the same memory consumption (\eg 
\cite{FanAKM14}). The issue is that we use the original element to compute its 
location when moving in the cuckoo hash, where in the dynamic case only its 
hash is available.

\section*{Acknowledgments}
We thank Ilan Komargodski for many helpful discussions, and we thank Yongjun Zhao for helpful discussions about the proof of 
Claim \ref{clm1}. We thank the anonymous referees for many insightful comments.

\bibliographystyle{amsalpha}
\bibliography{Main-AdversarialBloomFilter}
\appendix
\section{Preliminaries}
We start with some general notation. We denote by $[n]$ the set of numbers
$\{1,2,\dots,n\}$. We denote by $\negl\colon\N\to\R^+$ an arbitrary function 
$f\colon\N\to\R^+$ such that for all $c$ we have that $f(n) < 1/n^c$ for 
sufficiently large $n$. Finally, throughout this paper, we denote by $\log$ the 
base 2 logarithm. We use some cryptographic primitives, as defined in 
\cite{Goldreich2001}.

\subsection{Definitions}

\begin{definition}[One-Way Functions]\label{def:onewayfunctions}
  \label{def:OWF}
  A function $f$ is said to be \textsf{one-way} if the following holds:
  \begin{enumerate}
  \item There exists a polynomial-time algorithm $A$ such that $A(x)=f(x)$ for
    every $x\in\B^*$.
  \item For every probabilistic polynomial-time algorithm $A'$ and all
    sufficiently large $n$,
    \begin{align*}
      \Pr_{}[A'(1^n,f(x))\in f^{-1}(f(x))] < \negl(n),
    \end{align*}
    where the probability is taken uniformly over $x\in\B^n$ and
    the internal randomness of $A'$.
  \end{enumerate}
\end{definition}

\begin{definition}[Weak One-Way Functions]\label{def:weakOWF}
  A function $f$ is said to be \textsf{weakly one-way} if the following holds:
  \begin{enumerate}
  \item There exists a polynomial-time algorithm $A$ such that $A(x)=f(x)$ for
    every $x\in\B^*$.
  \item There exists a polynomial $p$ such that for every probabilistic
    polynomial-time algorithm $A'$,
    \begin{align*}
      \Pr_{}[A'(1^n,f(x))\in f^{-1}(f(x))] < 1-\frac{1}{p(n)},
    \end{align*}
    where the probability is taken uniformly over $x\in\B^n$ and
    the internal randomness of $A'$.
  \end{enumerate}
\end{definition}

We now define almost one-way functions, functions that are only hard to invert for infinitely many input lengths (compared with standard one-way functions that are hard to invert for all but finitely many input lengths).
\begin{definition}[Almost One-Way Functions]\label{def:almostOWF}
	A function $f$ is said to be \textsf{almost one-way} if the following 
	holds:
	\begin{enumerate}
		\item There exists a polynomial-time algorithm $A$ such that $A(x)=f(x)$ for
		every $x\in\B^*$.
		\item There exists a polynomial $p$ such that for every probabilistic
		polynomial-time algorithm $A'$ and for infinitely many $n$'s,
		\begin{align*}
		\Pr_{}[A'(1^n,f(x))\in f^{-1}(f(x))] < \frac{1}{p(n)},
		\end{align*}
		where the probability is taken uniformly over $x\in\B^n$ and
		the internal randomness of $A'$.
	\end{enumerate}
\end{definition}

\begin{definition}[Statistical Distance]\label{def:statisticaldistance}
  Let $X$ and $Y$ be two random variables with range $U$. Then the 
  \textsf{statistical distance} between $X$ and $Y$ is defined as
  \begin{align*}
	\Delta(X,Y) \triangleq \max_{A \subset U}\paren{\Pr[X \in A] - \Pr[Y \in 
	A]}
  \end{align*}
\end{definition}

\begin{definition}[Pseudorandom Functions (PRF)]\label{def:pseudorandom}
Let $\ell_1:\N \to \N$ and $\ell_2:\N \to \N$ be efficiently computable 
functions bounded by a polynomial, denoting the length of the domain and the 
range respectively. An efficiently computable family of functions
\begin{align*}
\mathcal{PRF} = \left \{\PRF_K:\B^{\ell_1(\lambda)} \to \B^{\ell_2(\lambda)} : 
K \in 
\B^{\lambda},\lambda \in \N \right \}
\end{align*}
is called a pseudorandom function if for every probabilistic polynomial time 
algorithm $A$ there exists a negligible function $\negl(\cdot)$ such that 
\begin{align*}
\left |\Pr_{K \gets \B^{\lambda}}[A^{\PRF_K(\cdot)}(1^\lambda)=1] - 
\Pr_{f \gets F_\lambda}[A^{f(\cdot)}(1^\lambda) = 1] 
\right | \le \negl(\lambda)
\end{align*}
where $F_\lambda$ is the set of all functions that map $\B^{\ell_1(\lambda)}$ 
into $\B^{\ell_2(\lambda)}$.
\end{definition}

\begin{definition}[Pseudorandom Permutation (PRP)]\label{def:pseudorandomP}
Let $\ell:\N \to \N$ be an efficiently computable 
function bounded by a polynomial, denoting the length of the domain. An 
efficiently computable family of {\em permutations}
\begin{align*}
\mathcal{PRP} = \left \{\PRP_K:\B^{\ell(\lambda)} \to \B^{\ell(\lambda)} : 
K \in 
\B^{\lambda},\lambda \in \N \right \}
\end{align*}
is called a pseudorandom permutation if for every probabilistic polynomial time 
algorithm $A$ there exists a negligible function $\negl(\cdot)$ such that 
\begin{align*}
\left |\Pr_{K \gets \B^{\lambda}}[A^{\PRP_K(\cdot)}(1^\lambda)=1] - 
\Pr_{p \gets P_\lambda}[A^{p(\cdot)}(1^\lambda) = 1] 
\right | \le \negl(\lambda)
\end{align*}
where $P_\lambda$ is the set of all permutations over $\B^{\ell(\lambda)}$.
\end{definition}

\begin{definition}[Universal Hash Family]\label{def:universal}
A family of functions $\HH \subset \{h \mid h:U \to [m]\}$ is called universal 
if for any $x_1,x_2 \in U$, $x_1 \ne x_2$: 
$$\Pr_{h \in \HH}[h(x_1) = h(x_2)] \le \frac{1}{m}$$
\end{definition}

\begin{lemma}\label{clm:logbinom}
For any $n,k \in \N$ it holds that $\logchoose{n}{k} \le k\log{n} - k\log{k} 
+ 2k$.
\end{lemma}
\begin{proof}
\begin{align*}
\logchoose{n}{k} \le \logp{\frac{n^k}{k!}} \le k\log{n} - \logp{k!} \le 
k\log{n} - k\log{k} + k\log{e} \le k\log{n} - k\log{k} + 2k.
\end{align*}
\end{proof}
\section{Implementing an Adversarial Resilient Bloom Filter}
\label{sec:implementation}
The original proposal of Bloom \cite{Bloom70} suggests a simple analysis 
and implementation. In practice, Bloom filters are used when performance is 
crucial, and extremely fast implementations are necessary. Thus, study of 
better implementation has attracted the community to further study this object. 
In particular, it has been shown that dictionary based implementations 
\cite{CarterFGMW78} outperform Bloom's construction not only by theoretical 
analysis 
but also in practical implementations (some examples of such implementations 
are \cite{BenderFJKKMMSSZ12,FanAKM14,PaghPR05,PutzeSS09}). We have defined and 
constructed an adversarial resilient Bloom filter under the 
assumption that pseudorandom functions exist. However, implementing our 
construction raises many difficulties: One needs to find 
implementations of the hash functions and the pseudorandom function that are 
fast on one hand, but are secure on the other.

We overcome these difficulties and provide an implementation of an adversarial 
resilient Bloom filter that is provably secure under the hardness of AES, and 
is essentially as fast as any other implementation of insecure Bloom filters. 
Our 
implementation uses AES-NI (Advanced Encryption Standard Instruction Set) that 
are embedded in most modern CPUs and provide hardware acceleration of the AES 
encryption and decryption algorithms. We use these instructions to implement 
our hashing functions and pseudorandom function. One advantage of 
blockcipher-based hashing is that they have been well studied and have 
several security proofs in the ideal cipher model.

The AES is used mainly to implement the secure pseudorandom function. However, 
to boost performance even more, we use the same function to implement the hash 
functions in the Bloom filter implementation as well. In particular, we used a 
linear probing based implementation. Since the AES results in a 128-bit block 
that is indistinguishable from a random block, we can split it 
into several parts and use each part for a different hash function. This 128-bit 
block suffices for about any realistic setting: supporting Bloom filters with 
more than $2^{80}$ elements and an error probability which can be as small as 
$2^{-40}$ simultaneously.

We have compared our implementation (in C\#) with several hash 
functions commonly used in Bloom filter implementations. We measured the time 
to insert $n$ elements, and then perform $n$ queries where half are in the set 
and half outside, for $n=10^5,10^6,10^7$, error $\varepsilon=2^{-16}$, and 
where the load of the dictionary was $\alpha=0.8$. The 
results are outlined in Figure \ref{fig:exp}. Notice that our hash function is 
only about 20-40 percent slower than the fastest one, for any choice of 
$n$, where the fastest functions have no security claims whatsoever. We believe 
that with better use of the AES-NI instruction (with a low-level language as 
C or C++) we could make the differences even smaller.

\begin{figure}[H]
\begin{center}
    \begin{tabular}{ | l | c | c | c |}
    \hline
    Hash Function & 100,000 & 1,000,000 & 10,000,000 \\ \hline
    MurmurHash3~\cite{applebymurmurhash3}	& 3.95ms&4.13ms&4.05ms \\ \hline
    CityHash~\cite{cityhash} 				&3.3ms&3.11ms&3.56ms \\ \hline
    Jenkins' Lookup3~\cite{Jenkinslookup3} 	&4.6ms&2.97ms&3.37ms\\ \hline
    BLAKE2~\cite{AumassonNWW13} 		& 49.3ms&44.4ms&44.4ms\\ \hline
    Pearson hashing~\cite{pearson1990fast} 	& 5.45ms&5.24ms&5.83ms \\ \hline
    This Work								& 4.65ms&4.54ms&4.86ms\\ \hline
    \end{tabular}

\caption{The performance time of the Bloom filter using different hash 
functions in milliseconds per $10^5$ operations (either insert or query) on a 
Intel core i7-2600K 3.40GHz with 8GB of memory.}
\label{fig:exp}
\end{center}
\end{figure}

One other advantage that only our construction enjoys is the ability to measure 
its performance on worst-case inputs. In general, the measurements in lab are 
performed on random inputs which can speed up the Bloom filter's query time. 
However, in practice, the inputs might be arbitrary and we might get worse 
results than measured in the lab. This is not the case with our implementation. 
Since we randomize the inputs using our implementation of a pseudorandom 
function, the performance on random inputs is indistinguishable from the 
performance on the worst possible sequence of inputs. Thus, if computing AES 
takes the same time for any element, we get that the performance in any real 
setting will be exactly as the performance measured in our experiments in the 
lab.

\end{document}